\documentclass[10pt,letterpaper]{article}

\newcommand*{\base}{.}
\usepackage[utf8]{inputenc}
\usepackage{hyperref}
\usepackage{lmodern}
\usepackage{amsmath, amsthm, amssymb}
\usepackage{accents}
\usepackage{microtype}
\usepackage[hmargin=3.3cm,vmargin=2.5cm]{geometry}
\usepackage[usenames]{xcolor}
\usepackage{nicefrac}
\usepackage{bbm}
\definecolor{Blue}{RGB}{5, 55, 77}
\definecolor{Red}{RGB}{20,50,10}
\usepackage{pgfplots}
\usepackage{authblk}
\usepackage {tikz}
\usetikzlibrary {positioning}
\usepackage{float}
\usepackage{aliascnt}
\usepackage{multirow}
\hypersetup{
    colorlinks,
    citecolor=Red,
    filecolor=Blue,
    linkcolor=Blue,
    urlcolor=Blue,
    breaklinks=false,
}
\usepackage{enumerate}
\usepackage{sectsty}
\usepackage{mathtools}
\usepackage[anythingbreaks]{breakurl}

\sectionfont{\large}
\subsectionfont{\normalsize}
\newcommand{\setbuilder}[2]{\left\{#1~\middle|~#2\right\}}
\newcommand{\set}[1]{\left\{#1\right\}}

\newcommand{\emid}{\,|\,}

\newcommand{\E}{\mathop{{}\mathbb{E}}}

\newcommand{\dist}{\operatorname{dist}}
\newcommand{\supp}{\operatorname{supp}}

\DeclareMathOperator*{\I}{\mathbf{I}}

\DeclareMathOperator{\D}{\mathbf{D}}

\newcommand{\muti}[2]{\I{\mkern-3mu}\left(#1:#2\right)}
\newcommand{\kldiv}[2]{\D{\mkern-5mu}\left(#1\,\middle\|\,#2\right)}
\newcommand{\kldivs}[2]{\D(#1\,\|\,#2)}

\newcommand{\ryent}[2]{\mathrm{H}_{#1}\left(#2\right)}

\DeclareMathOperator{\rank}{rank}

\newcommand{\lone}[1]{\left\|#1\right\|_1}
\newcommand{\lnorm}[2]{\left\|#1\right\|_{#2}}
\newcommand{\iprod}[2]{\left\langle#1,#2\right\rangle}

\newcommand{\transpose}{\intercal}
\input{\base/definitions/theoremlike}
\usepackage[
  backend=biber,
  style=alphabetic,
  natbib=true,
  backref,
  url=true]{biblatex}
\DefineBibliographyStrings{english}{%
  backrefpage = {cited on page},
  backrefpages = {cited on pages},
}
\newcommand{\field}{\mathbb F}
\newcommand{\reals}{\mathbb R}
\newcommand{\realspos}{\mathbb R_+\!}

\newcommand{\binary}{\{0,1\}}
\newcommand{\cube}{\{0,1\}^n}

\usepackage{expl3}
\ExplSyntaxOn
\newcommand \breakDOI[1]
 {
  \tl_map_inline:nn { #1 } { \href{http://dx.doi.org/#1}{##1} \penalty0 \scan_stop: }
 }
\ExplSyntaxOff

\DeclareFieldFormat{doi}{%
  \mkbibacro{DOI}\addcolon\space
  {\breakDOI{#1}}}

\begin{document}
\author{\fontsize{11}{11}\selectfont Mert Sağlam\\%
\fontsize{10}{10}\selectfont saglam@uw.edu}
\title{\bf\fontsize{15}{19}\selectfont
Near log-convexity of measured heat in
(discrete) time and consequences}
\date{}
\maketitle
\vspace{-25px}

% !TeX root = heatdiscrete.tex
\begin{abstract}
Let $u,v \in \reals^\Omega_+$ be positive
unit vectors and
$S\in\reals^{\Omega\times\Omega}_+$ be a
symmetric substochastic matrix.
For an integer $t\ge 0$, let 
$m_t = \smash{\iprod{v}{S^tu}}$, which
we view as the
heat measured by $v$ after an initial heat
configuration $u$ is let to diffuse for $t$
time steps according to $S$. 
Since $S$ is entropy improving,
one may intuit that $m_t$ should not change too
rapidly over time. We give the following formalizations
of this intuition.

We prove that
$m_{t+2} \ge m_t^{1+2/t}\!,$ an inequality
studied earlier by Blakley and Dixon 
(also Erdős and Simonovits) for $u=v$
and shown true under the restriction $m_t\ge e^{-4t}$.
Moreover we prove that for any $\eps>0$, 
a stronger inequality 
$m_{t+2} \ge t^{1-\eps}\cdot \smash{m_t^{1+2/t}}$
holds unless
$m_{t+2}m_{t-2}\ge \delta m_t^2$
for some $\delta$ that depends on $\eps$ only.
Phrased differently, $\forall \eps> 0, 
\exists \delta> 0$ such that $\forall S,u,v$
\begin{equation*}
\frac{m_{t+2}}{m_{t}^{1+2/t}}\ge
\min\set{t^{1-\eps}, \delta\frac{m_t^{1-2/t}}{m_{t-2}}},
\quad \forall t \ge 2,
\end{equation*}
which can be viewed as a truncated 
log-convexity statement. 

Using this inequality, we answer two related
open questions in complexity theory: Any
property tester for $k$-linearity requires
$\Omega(k\log k)$ queries and the randomized
communication complexity of the $k$-Hamming
distance problem is $\Omega(k\log k)$.
Further we show that any randomized parity
decision tree computing $k$-Hamming weight has size
$\exp\nparen{\Omega(k\log k)}$.
\end{abstract}

% !TeX root = heatdiscrete.tex
\section{Introduction}
\label{sec:introduction}

Suppose that some initial heat configuration 
$u\colon\Omega\to\realspos$ is given over a finite 
space $\Omega$ and the configuration evolves 
according to the map $w\mapsto Sw$ in each time step $t=0,1,\ldots$, for some symmetric stochastic matrix 
$S\colon \Omega\times\Omega\to\realspos\,$.
Assume that we are interested in the amount of heat contained in a certain region $R\subseteq \Omega$ 
and how this quantity changes over time. In notation, 
assuming $\lnorm{u}{2}=1$ for normalization purposes 
and $v(x) \defeq \indicate{x\in R}/|R|^{1/2}$ for 
$x\in\Omega$, we would like to understand how
\begin{align*}
m_t \defeq \iprod{v}{S^tu}
\end{align*}
changes as a function of $t$. In this paper we 
derive local bounds that $\set{m_t}_{t=0}^\infty$ 
must obey for any $S,u$ and $v$ satisfying the 
symmetry, magnitude and positivity constraints above 
(in fact our bounds work for any countable
$\Omega$, arbitrary non-negative unit vector $v$ 
and symmetric non-negative $S$).
%The first of our bounds concludes a study initiated in \cite{BlakleyD1966} by Blakley and Dixon.
Further, we establish a tight connection between such bounds and the well-studied $k$-Hamming distance problem
\cite{PangG1986, Yao2003, CormodePS2000, BarYossefJKK2004,
GavinskyKW2004, HuangSZZ2006,
BlaisBM2012, BuhrmanGMW2012, BlaisBG2014, AmbainisGSU2015} 
and the $k$-Hamming weight problem
\cite{AdaFH2012, BlaisK2012, BuhrmanGMW2012}
and obtain the first tight bounds for respectively
the communication complexity and parity decision tree 
complexity of them.

Our tight $\Omega(k \log (k/\delta))$ lower bound 
for the $\delta$-error communication complexity of 
the $k$-Hamming distance problem (that applies 
whenever $k^2< \delta n$) answers affirmatively a 
conjecture stated in \cite{BlaisBG2014} 
(Conjecture 1.4).
Prior to our work, the best impossibility results 
for this problem were an $\Omega(k\log^{(r)}k)$ bits 
lower bound ($\smash{\log^{(r)}}$ being the iterated 
logarithm) that applies to any randomized $r$-round 
communication protocol \cite{SaglamT2013}, 
and an $\Omega(k\log (1/\delta))$ lower bound that 
applies to any $\delta$-error randomized protocol for 
$k <\delta n$ \cite{BlaisBG2014}.

Our parity decision tree lower bound shows that any 
$\delta$-error parity decision tree solving the 
$k$-Hamming weight problem has size 
$\exp\Omega\nparen{k\log (k/\delta)}$, 
which directly implies an $\Omega(k\log (k/\delta))$ 
bound on the depth of any such decision tree. 
Previously no nontrivial lower bound was known for 
the parity decision tree size of this problem and 
an $\Omega(k\log (1/\delta))$ bound on the parity 
decision tree depth followed from the communication 
complexity bound of \cite{BlaisBG2014}. Prior to \cite{BlaisBG2014}, the best bound on the parity decision tree depth was $\Omega(k)$, derived in 
\cite{BlaisBM2012} and \cite{BlaisK2012}.

Either by combining our communication complexity lower bound with the reduction technique developed in \cite{BlaisBM2012} or by combining
our parity decision tree lower bound with a 
reduction given in \cite{BhrushundiCK2014}, 
one obtains an 
$\Omega(k\log (k/\delta))$ bound for any (potentially adaptive) property tester for the $\delta$-error probability $k$-linearity testing problem. This establishes the correct bound for this problem which was studied extensively 
\cite{FischerLNRRS2002, Goldreich2010, BlaisK2012, BhrushundiCK2014,BuhrmanGMW2012, BlaisBM2012}
since \cite{FischerLNRRS2002} or earlier.

\subsection{Motivating our bounds on $m_t$}
We would like provide some intuition as to why one should expect
\begin{align}
m_{t+2}    &\ge m_t^{1+2/t}, \text{ and}\label{eq:bd66}\\
m_{t+2}    &\ge m_t^{1+2/t}\cdot \min\set{t^{1-\eps}, \delta\frac{m_t^{1-2/t}}{m_{t-2}}}
\end{align}
to hold for appropriate $\eps,\delta$.
Recalling that $S$ is a symmetric matrix with maximum 
eigenvalue 1, we may write $S=QDQ^\transpose$
for an orthonormal matrix $Q$ having columns $q_x$, $x\in\Omega$
and a diagonal matrix $D$ with entries $\lambda_x \le 1$, 
$x\in \Omega$. Plugging this into 
$m_t=\iprod{v}{S^tu}$, we get
\begin{align}
\label{eq:spectralsum}
m_t = \sum_{x\in\Omega}\lambda_x^t\iprod{u}{q_x}\iprod{v}{q_x}.
\end{align}
For sake of analogy let us drop our assumption that $S, u, v$
are coordinate-wise nonnegative for a moment but instead assume
that each summand in the right hand side of
\autoref{eq:spectralsum}
is nonnegative by some coincidence. In this case we can
consider $\set{m_t}_{t=0}^\infty$ as the moment sequence of a 
random variable supported on $[0,1]$ that takes the
value $|\lambda_x|$ with probability 
$|\iprod{u}{q_x}\iprod{v}{q_x}|$ and the value 0 with
probability
$1-\sum_x|\iprod{u}{q_x}\iprod{v}{q_x}|$ (which is nonnegative
by Cauchy-Schwarz inequality). 
This would imply that $\set{m_t}_{t=0}^\infty$ is 
{\em completely monotone} by Hausdorff's characterization 
\cite{Hausdorff1921} and therefore log-convex 
(e.g., \cite{NiculescuP2005}, Section 2.1, Example 6).

One particular implication of the log-convexity of 
$\set{m_t}_{t=0}^\infty$, that 
$\frac{1}{t}\log m_t + \frac{t-1}{t}\log m_0\ge \log m_1$,
when combined with the fact $0\le m_0\le 1$ (that follows from our
assumption on the terms of \autoref{eq:spectralsum}), leads to
$m_t\ge m_1^t$.
In 1958, Mandel and Hughes showed that if $u=v$, rather
surprisingly, one can trade the assumption that the summands of
\autoref{eq:spectralsum} are nonnegative with the assumption
that $S$ and $u=v$ are coordinate-wise nonnegative and still
obtain the conclusion $m_t\ge m_1^t$:
\begin{theorem}[Mandel and Hughes \cite{MandelH1958}]
\label{thm:blakley-roy}
Let $u$ be a nonnegative unit vector and 
$S$ be a symmetric matrix with nonnegative entries. For an integer%
\footnote{Since $u=v$ here, the summands inside 
\autoref{eq:spectralsum} are nonnegative when $t$ is even  
so this theorem is most interesting for $t$ odd.}
$t\ge1$ we have $\iprod{u}{S^tu}\ge \iprod{u}{Su}^t$.
\end{theorem}

A more general implication of the log-convexity of 
$\set{m_t}_{t=0}^\infty$ and that $m_0\le 1$ is that for 
$k\ge t$, $\frac{t}{k}\log m_k + \frac{k-t}{k}\log m_0\ge \log m_t$,
therefore $m_k^t\ge m_t^k$. In 1966,
Blakley and Dixon \cite{BlakleyD1966} investigated 
whether $m_k^t\ge m_t^k$ holds in the case $u=v$ when the nonnegativity
assumption on the summands of \autoref{eq:spectralsum} is
replaced by the coordinate-wise nonnegativity of $S$, $u=v$.
They note that the inequality $m_k^t\ge m_t^k$ 
fails when $k$ and $t$ have different 
parity and otherwise holds true under the restriction 
$m_t\ge e^{-4t}$.
While the following is not explicitly stated as a conjecture 
in \cite{BlakleyD1966}, they write
\begin{quote}
if $t> 1$, [...] we cannot show that the inequality
\autoref{eq:bd66} holds for each nonnegative $|\Omega|$-vector
$u$ if $S$ is nonnegative.
\end{quote}
so with the earlier caveat we attribute the following to Blakley and
Dixon \cite{BlakleyD1966}:

\begin{conjecture}[Blakley and Dixon \cite{BlakleyD1966}]
\label{conj:blakley-dixon}
Let $S\colon\Omega\times\Omega\to\realspos$ be a symmetric matrix
with nonnegative entries and let $u\colon\Omega\to\realspos$ be a
nonnegative unit vector. For positive integers 
$k\ge t$ of the same parity, we have
  \begin{align*}
    \iprod{u}{S^ku}^t\ge \iprod{u}{S^tu}^k.
  \end{align*}
\end{conjecture}

In \autoref{sec:monotonicity} we prove the following theorem which
shows that a generalization of \autoref{conj:blakley-dixon}
holds true.
\begin{theorem}
\label{thm:blakley-dixon}
Let  $S\colon\Omega\times\Omega\to\realspos$ be a symmetric matrix
with nonnegative entries and 
$u,v\colon\Omega\to\realspos$ be
nonnegative unit vectors. For positive integers 
$k\ge t$ of the same
parity, we have
  \begin{align*}
    \iprod{v}{S^ku}^t\ge \iprod{v}{S^tu}^k.
  \end{align*}
\end{theorem}
It goes without saying that \autoref{eq:bd66} is 
equivalent to \autoref{thm:blakley-dixon} as we can rearrange 
\autoref{eq:bd66} to $\smash{m_{t+2}^{1/(t+2)}\ge m_t^{1/t}}$ and apply
it iteratively to obtain 
$\smash{m_k^{1/k}\ge\cdots\ge m_{t+2}^{1/(t+2)}\ge m_t^{1/t}}$ whenever 
$k\ge t$ and $k,t$ have the same parity. Moreover, while defining 
\autoref{eq:bd66} we assumed $S$ to be substochastic only to illustrate our
interpretation of the inequality: indeed any nonnegative $S$ can be scaled 
to be substochastic as both sides of \autoref{eq:bd66} are 
$(t+2)$-homogeneous in $S$.

In \autoref{thm:blakley-roy} and \autoref{thm:blakley-dixon} we observed that increasingly more general implications of the log-convexity of 
$\set{m_t}_{t=0}^\infty$ can be derived by only assuming the coordinate-wise 
nonnegativity of $S,u$ and $v$. One may naturally wonder if the coordinate-wise 
nonnegativity of $S,u$ and $v$ implies
the log-convexity of 
$\set{m_t}_{t=0}^\infty$ in its entirety. Unfortunately the following
example shows that this is far from the truth.
\begin{figure}[H]
\centering
\begin {tikzpicture}[]
\tikzstyle{state}=[circle, draw, inner sep=0pt, minimum size=27pt]
\node[state] (A){$0$};
\node[state] (B) [right =of A] {$1$};
\node[state] (C) [right =of B] {$2$};
\node        (E) [right =of C] {$\cdots$};
\node[state] (J) [right =of E] {$t-1$};
\node[state] (K) [right =of J] {$t$};
\path (A) edge node[below] {$\eps$} (B);
\path (B) edge node[below] {$\eps$} (C);
\path (C) edge node[below] {$\eps$} (E);
\path (E) edge node[below] {$\eps$} (J);
\path (J) edge node[below] {$\eps$} (K);
\end{tikzpicture}
\caption{$\Omega=\set{0,1,\ldots, t}$, $S(i,i+1)=S(i+1,i)=\eps$ for $i=0,\ldots,t-1$.}
\label{fig:ex1}
\end{figure}
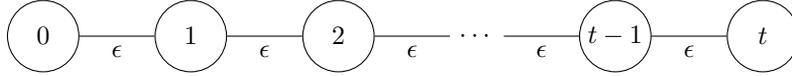
Consider the transition matrix $S$ on $\Omega=\set{0,1,\ldots,t}$ such that 
$S(i,i+1)=S(i+1,i)=\eps$ for $i=0,1,\ldots,t-1$ and $S(i,j)=0$ otherwise.
Let $u$ and $v$ be the point masses respectively on states $0$ and $t$; 
namely $u=[1,0,\ldots,0]^\transpose$ and $v=[0,0,\ldots,1]^\transpose$.
We have $m_{t-2}=0$, $m_t=\eps^t$ and $m_{t+2}=t\eps^{t+2}$. Therefore
$m_{t-2}m_{t+2} = 0 \not\ge \eps^{2t} = m_t^2$.
In this example the log-convexity breaks (in the strongest possible way) 
because the states $0$ and $t$ are separated by $t$ hops according to $S$ and the 
point mass at state $0$ cannot reach state $t$ before the $t$th time step. 

Our next theorem shows that such reachability issues are essentially 
the only way the log-convexity property can fail to hold:
\begin{theorem}
\label{thm:main}
For every $\eps>0$ there is a $\delta>0$ such that for any symmetric matrix 
$S\colon\Omega\times\Omega\to\realspos$ and unit vectors 
$u,v\colon\Omega\to\realspos$ with nonnegative entires, defining $m_t$ as before, we have
\begin{align}
\frac{m_{t+2}}{m_{t}^{1+2/t}}\ge
\min\set{t^{1-\eps}, \delta\frac{m_t^{1-2/t}}
{m_{t-2}}},\quad\quad\forall t\ge 2.
\label{eq:insidemain}
\end{align}
\end{theorem}
In other words, \autoref{thm:main} shows that one can recover a
truncated version of the log-convexity of
$\set{m_t}_{t=0}^\infty$ from just the coordinate-wise
nonnegativity assumption of $S,u$ and $v$.
We stress that \autoref{thm:main} is tight 
up to the appearance of $\eps$ and the choice of 
$\delta=\delta(\eps)$. A direct calculation
on \autoref{fig:ex1} for time steps $t, t+2, t+4$ 
shows that \autoref{eq:insidemain} cannot be improved to
\begin{align*}
\frac{m_{t+2}}{m_{t}^{1+2/t}}\ge
\min\set{t^{1-2/t}, 
\nparen{\frac{1+\eta}{2}}\frac{m_t^{1-2/t}}{m_{t-2}}}
\end{align*} for $\eta>0$.

\subsection{Related work on $m_t$}
Almost simultaneously with the work of Mandel and Hughes
\cite{MandelH1958}, Mulholland and Smith also prove
\autoref{thm:blakley-roy} in \cite{MulhollandS1959} and moreover they characterize the equality conditions of the inequality. Independently, in 1965, Blakley and Roy \cite{BlakleyR1965} prove the same inequality and
characterize the equality conditions and 
\cite{London1966} provides an alternative proof to that 
of \cite{MulhollandS1959} in 1966.
We remark that \autoref{thm:blakley-roy} is 
most commonly referred to as the Blakley-Roy bound
or ``Sidorenko's conjecture for paths''.
Note these results show that \autoref{conj:blakley-dixon} 
is true whenever $t$ divides $k$. 
Finally in 2012, Pate shows that $m_t\ge m_1^t$ 
without the restriction $u=v$:

\begin{theorem}[Pate \cite{Pate2012}]
Let $S\colon\Omega\times\Omega\to\realspos$ be a 
symmetric matrix
with nonnegative entries and let 
$u,v\colon\Omega\to\realspos$ be
nonnegative unit vectors.
It holds that
\begin{align*}
\iprod{v}{S^{2t+1}u} \ge \iprod{v}{Su}^{2t+1},
\end{align*}
with equality if and only if 
$\iprod{v}{S^{2t+1}u} =0$ or 
$Su=\lambda v$ and 
$Sv=\lambda u$ 
for some $\lambda\in\realspos$.
\end{theorem}
This result already shows that 
\autoref{thm:blakley-dixon} is true when $t$ divides $k$ 
but such a bound does not have any implications for our applications in complexity theory. In \cite{ErdosS1982},
Erdős and Simonovits conjecture the following.
\begin{conjecture}
[Erdős and Simonovits \cite{ErdosS1982}, Conjecture 6]
\label{conj:erdos-simonovits}
For a graph $G=(V,E)$, let $w_k(G)$ be the number length $k$ walks in $G$ divided by $|V|$.
For an undirected graph $G$, we have $w_k(G)^t\ge w_t(G)^k$ for $k>t$ of the same parity.
\end{conjecture}
This conjecture was recalled in a recet book
\cite{Taubig2017} by Täubig as Conjecture~4.1.
Note that \autoref{conj:erdos-simonovits} is a specialization of \autoref{conj:blakley-dixon} to
$S$ having 0-1 entries and  
$u=\mathbf{1}/\sqrt{|V|}$ therefore our \autoref{thm:blakley-dixon}
verifies \autoref{conj:erdos-simonovits} as well.

\subsection{Our results in complexity theory}
\label{sec:results-cc}
Here we list our results in complexity theory; see
\autoref{sec:applications} for the definition
of the models and the problems.
The following theorem 
(which was already known \cite{BlaisBG2014})
is a consequence of \autoref{thm:blakley-dixon} and uses the 
standard corruption technique in 
communication complexity.
\begin{theorem}
\label{thm:klogdelta}
Any two party $\delta$-error randomized protocol
solving the $k$-Hamming distance problem 
over length-$n$ strings communicates 
at least $\Omega(k\log (1/\delta))$ 
bits for $k^2\le \delta n$.
\end{theorem}
The next is our main result for the communication complexity of the $k$-Hamming distance problem
and is a consequence of \autoref{thm:main}. 
This result cannot be obtained by the standard 
corruption technique and
requires a suitable modification similar to \cite{Sherstov2012}.
\begin{theorem}
\label{thm:klogk}
Any two party $\delta$-error randomized protocol
solving the $k$-Hamming distance problem 
over length-$n$ strings communicates 
at least $\Omega(k\log (k/\delta))$ 
bits for $k^2\le \delta n$.
\end{theorem}
\begin{theorem}
\label{thm:paritysize}
Any $\delta$-error parity decision tree 
deciding the $k$-Hamming weight predicate 
over length-$n$ strings
has size $\exp \Omega\nparen{k\log (k/\delta)}$
for $k^2< \delta n$.
\end{theorem}
\begin{corollary}
\label{cor:propertytest}
Any $\delta$-error probability 
property tester for 
$k$-linearity requires 
$\Omega(k\log (k/\delta))$ queries.
\end{corollary}

Note the bound $m_k^t\ge m_t^k$ obtained in \cite{BlakleyD1966} under the condition $m_t\ge e^{-4t}$ 
does not have any implications for the communication complexity of the $k$-Hamming distance problem as our reduction crucially uses the fact that $u$ and $v$ are arbitrary, however it does lead to an 
$\exp\Omega(k)$ lower bound for the parity decision tree size of the $k$-Hamming weight problem when combined 
with our reduction.

\begin{remark}
Note that in \autoref{thm:blakley-dixon}, when $u=v$ and either
both $k,t$ are even or $S$ is positive semidefinite, the
summands of \autoref{eq:spectralsum} become nonnegative and the
inequality holds trivially. For our application in communication
complexity we crucially use the fact that $u$ and $v$ are
arbitrary and for our application in parity decision trees, one
can do away with $u=v$ but only at the expense of having to
choose $k,t$ odd. 
In both results the $S$ we choose has eigenvalues $1$ and $-1$ with equal multiplicities and therefore far away from
being positive semidefinite.
In either case, the implications of
\autoref{thm:blakley-dixon} in complexity theory follow 
from the interesting cases of this theorem.
\end{remark}

% !TeX root = heatdiscrete.tex
\section{Preliminaries}
\label{sec:preliminaries}

We denote by $[n]$ the set $\set{1,2,\ldots,n}$. We
take $\exp$ and $\log$ functions to the base 2.
Let $\Omega$ be a countable set. 
For a function $\mu\colon\Omega\to\realspos$ and a set 
$\Psi\subseteq\Omega$, we use the shorthand
\begin{align*}
\mu(\Psi) \defeq \sum_{x\in\Psi}\mu(x).
\end{align*}
A function $\mu\colon\Omega\to\realspos$
is said to be a distribution on $\Omega$ if
$\mu(\Omega) = 1$ and a subdistribution
if $\mu(\Omega) \le 1$. For a function
$\mu$ on $\Omega$, we define
\begin{align*}
\supp(\mu)\defeq \setbuilder{x\in \Omega}{\mu(x)>0}.
\end{align*}
For two distributions $\mu\colon\Omega_1\to\realspos$
and $\nu\colon\Omega_2\to\realspos$, let us denote by
$\mu\nu$ the distribution on $\Omega_1\times\Omega_2$
given by $(\mu\nu)(x_1,x_2) = \mu(x_1)\nu(x_2)$.

For a discrete random variable $X$, we denote by
$\dist(X)$ the distribution function of $X$ and
we define $\supp(X)\defeq\supp(\dist(X))$. 
If $X$ is so that $\dist(X)\colon\Omega\to\realspos$,
then we say that $X$ has sample space $\Omega$.
Two 
random variables $X$ and $Y$ are said to be 
independent if $\dist(XY) =\dist(X)\dist(Y)$.

\begin{lemma}[Jensen \cite{Jensen1906}, Formula (5)]
\label{lem:jensen}
Let $X$ be a real-valued random variable and
$f$ be a convex function. We have
$\E\sparen{f(X)} \ge f(\E[X])$. When
$f$ is strictly convex, the inequality holds
with equality if and only if $X$ is constant
with probability 1.
\end{lemma}

\subsection{Information theory}
\label{sec:info}
In this section we review the definitions and 
facts we use from information theory.
Let $\mu$ and $\nu$ be two 
nonnegative functions on $\Omega$.
The Kullback-Leibler divergence \cite{Wald1945, KullbackL1951}
of $\mu$ from $\nu$,
denoted $\kldiv{\mu}{\nu}$, is defined by
\begin{align}
\label{eq:kldiv}
  \kldiv{\mu}{\nu} \defeq \sum_{x\in \Omega}
      \mu(x)\log\frac{\mu(x)}{\nu(x)}\,.
\end{align}
Here, if $\mu(x)=0$ for some $x$, then its contribution
to the summation is taken as 0, even when $\nu(x)=0$.
The divergence is undefined if there is an
$x\in \Omega$ such that $\mu(x)>0$ and
$\nu(x)=0$.
It can be shown that if the related series 
converges for the right hand side of \autoref{eq:kldiv},
it converges absolutely, 
which justifies leaving the summation order unspecified.
A fundamental property of 
$\kldiv{\cdot}{\cdot}$ is that
the divergence of a distribution from a subdistribution 
is always nonnegative.
\begin{lemma}[Gibbs \cite{Gibbs1902}, Theorem VIII]
\label{lem:gibbs}
Let $\mu,\nu\colon\Omega\to\reals$ be such that $\mu$ is a
distribution and $\nu$ is a subdistribution.
We have $\kldiv{\mu}{\nu}\ge 0$ with equality if and
only if $\mu=\nu$.
\end{lemma}
\begin{proof}
This follows from the next lemma by noting that 
$-\log \nu(\Omega)\ge 0$ for $\nu$ a 
subdistribution.
\end{proof}
\begin{lemma}[Kullback and Leibler \cite{KullbackL1951},
Lemma 3.2]
\label{lem:klcond}
Let $\mu,\nu\colon\Omega\to\realspos$ be so that
$\mu$ is a distribution on $\Omega$ and 
$\supp(\mu) = \Psi\subseteq\Omega$. We have
\begin{align*}
\kldiv{\mu}{\nu}\ge -\log \nu(\Psi)
\end{align*}
with equality if an only if $\mu(x) = \nu(x)/\nu(\Psi)$ for
$x\in\Psi$ and $\mu(x) = 0$ for $x\notin\Psi$.
\end{lemma}
\begin{proof}
By \autoref{eq:kldiv} we write
\begin{align*}
\kldiv{\mu}{\nu} &= -\sum_{x\in \Psi}
    \mu(x)\log\frac{\nu(x)}{\mu(x)}
  \ge -\log\sum_{x\in \Psi}\nu(x) =-\log \nu(\Psi)\,,
\end{align*}
where the inequality follows from
\autoref{lem:jensen} and concavity of
$z\mapsto\log z$ on $\realspos$. 
If $\mu(x) = \nu(x)/\nu(\Psi)$ for $x\in\Psi$, we have 
$\kldiv{\mu}{\nu}=-\log \nu(\Psi)$ by direct computation.
Otherwise $\kldiv{\mu}{\nu}>-\log\nu(\Psi)$ by strict 
concavity of $z\mapsto\log z$.
\end{proof}

We extend the divergence notation
$\kldiv{\cdot}{\cdot}$ to apply to random
variables as follows. Let $X,Y$ be discrete random
variables on the same sample space $\Omega$. Define
\begin{align}
\label{eq:klrv}
\kldiv{X}{Y}\defeq \kldiv{\dist(X)}{\dist(Y)}.
\end{align}
With this notation in hand, we are ready to define the 
conditional divergence. Let $X_1X_2$ and $Y_1Y_2$ be 
random variables defined on the sample space $\Omega_1\times \Omega_2$.
The divergence of $X_1\emid X_2$ from $Y_1\emid Y_2$ is
defined by 
\begin{align}
\kldiv{X_1\emid X_2}{Y_1\emid Y_2}
    \defeq \E_{x_2\sim X_2}
        \kldiv{X_1\emid X_2=x_2}{Y_1\emid Y_2 = x_2}.
\end{align}
Here, for each $x_2\in\supp(X_2)$, 
$X_1\emid X_2=x_2$ and $Y_1\emid Y_2 = x_2$ 
are random variables on the sample space
$\Omega_1$ obtained from, respectively
$X_1X_2$ and $Y_1Y_2$, by conditioning on the
second coordinate equaling $x_2$.

\begin{lemma}[e.g., \cite{CoverT2006}]
\label{lem:klchain}
Let $X_1X_2$ and $Y_1Y_2$ be random variables, both on the
sample space $\Omega_1\times \Omega_2$. We have
\begin{align*}
\kldiv{X_1X_2}{Y_1Y_2} = 
    \kldiv{X_1}{Y_1} + \kldiv{X_2\emid X_1}{Y_2\emid Y_1}.
\end{align*}
\end{lemma}
\iffalse
\begin{proof}
Let $\mu,\nu\colon\Omega_1\times \Omega_2\to \realspos$
be the distributions of respectively $X_1X_2$ and $Y_1Y_2$.
Using the shorthands
$\mu(\Omega_1,x_2)\defeq\sum_{x_1\in\Omega_1}\mu(x_1,x_2)$
and 
$\nu(\Omega_1,x_2)\defeq\sum_{x_1\in\Omega_1}\nu(x_1,x_2)$,
we write
\begin{align*}
  \kldiv{X_1\emid X_2}{Y_1\emid Y_2} &=
     \sum_{x_2\in\Omega_2}\mu(\Omega_1,x_2)
       \sum_{x_1\in\Omega_1}\frac{\mu(x_1,x_2)}
           {\mu(\Omega_1,x_2)}
       \log\frac{\mu(x_1,x_2)\nu(\Omega_1,x_2)}
                {\nu(x_1,x_2)\mu(\Omega_1,x_2)}\\
  &= \sum_{x_1,x_2}\mu(x_1,x_2)
       \log\frac{\mu(x_1,x_2)}{\nu(x_1,x_2)} +
     \sum_{x_2}
       \mu(\Omega_1,x_2)\log\frac{\nu(\Omega_1,x_2)}
           {\mu(\Omega_1,x_2)}\\
\intertext{by splitting the terms inside the logarithm. 
Using \autoref{eq:kldiv} together with \autoref{eq:klrv}
 we conclude}
  &= \kldiv{X_1X_2}{Y_1Y_2} - \kldiv{X_2}{Y_2}.
\end{align*}
Rearranging, we obtain the statement of the lemma.
\end{proof}
\fi
\paragraph{The mutual information}
Let $X$ and $Y$ be jointly distributed random variables.
The 
mutual information of $X$ and $Y$, 
denoted $\muti{X}{Y}$, is defined as
\begin{align}
\muti{X}{Y} \defeq \kldiv{\dist(X,Y)}{\dist(X)\dist(Y)}.
\label{eq:muti-def}
\end{align}
The mutual information of a random variable 
with itself, i.e., the quantity $\muti{X}{X}$ is 
called the Shannon entropy of $X$.
% !TeX root = heatdiscrete.tex
\section{Monotonicity of 
$t\mapsto m_{2t}^{1/(2t)}$ and 
$t\mapsto m_{2t+1}^{1/(2t+1)}$}
\label{sec:monotonicity}

Let $S\colon\Omega\times\Omega\to\realspos$ be a 
symmetric matrix with nonnegative entries, 
$u,v\colon\Omega\to\realspos$ nonnegative unit vectors  
and $m_t = \iprod{v}{S^tu}$.
In this section we prove \autoref{thm:blakley-dixon}
which we restate here 
(with additional equality conditions) 
for the convenience of the reader.
Recall this theorem confirms 
\autoref{conj:blakley-dixon} 
and \autoref{conj:erdos-simonovits}.

\begingroup
\def\thetheorem{\ref{thm:blakley-dixon}}
\begin{theorem}[restated]
Let $S\colon\Omega\times\Omega\to\realspos$ be a 
symmetric matrix with nonnegative entries and 
$u,v\colon\Omega\to\realspos$ be
nonnegative unit vectors. For positive integers 
$k\ge t$ of the same
parity, we have
  \begin{align}
    \iprod{v}{S^ku}^t\ge \iprod{v}{S^tu}^k,
    \label{eq:bd-inthm}
  \end{align}
with equality if and only if $\iprod{v}{S^ku} = 0$ or
$Su =\lambda v$ and $Sv=\lambda u$ for some 
$\lambda\in\realspos$ when $t$ is odd and 
$u=v$ is an eigenvector of $S^2$ when $t$ is even.
\end{theorem}
\addtocounter{theorem}{-1}
\endgroup

We prove \autoref{thm:blakley-dixon} by an 
information theoretic argument.
Define the distributions $\mu \defeq u/\lone{u}$ and 
$\nu \defeq v/\lone{v}$.
Since either side of 
\autoref{eq:bd-inthm} is $kt$-homogeneous in $S$, we may
assume that $S$ is substochastic by scaling as needed. 
Having fixed this normalization, we view 
\autoref{eq:bd-inthm} as a statement about random walks
on $\Omega$ that start from a state sampled according to 
$\mu$ or $\nu$ and evolve according to the transition 
matrix $S$.

\subsection{Reference random walks}
\label{sec:refwalk}
\def\OC{\Omega_\circ}
Let $\OC = \Omega\cup\set{r}$ for some state $r\notin\Omega$ and 
$t$ be a positive integer. Recall that 
$\mu = u/\lone{u}$ and $\nu = v/\lone{v}$.
We start by defining random walks $F^t, B^t$ on $\OC$
that evolve in discrete time steps $-1, 0,1, \ldots, t, t+1$.

The random walk $F^t$ starts at $r$ and transitions to 
a state $x\in\Omega$ with probability $\mu(x)$ at time step $-1$.
In steps $0,1,\ldots, t-1$, the random walk proceeds according to
the transition matrix $S$. At the time step $t$,
each state $x\in\Omega$ transitions to $r$ 
with probability $\nu(x)$ and
transitions to an arbitrary state in 
$\Omega$ with probability
$1-\nu(x)$ 
(say, all of them to the same arbitrary state).
We view $F^t$ as a joint random variable
$F^t = (F^t_{-1},F^t_0,\ldots,F^t_{t+1})$,
where $F^t_i$ is the location of the walk in time step $i$.

The random walk $B^t$ proceeds backwards in time. At time 
step $t+1$ the walk $B^t$ starts at $r$ and transitions
to a state $x\in\Omega$ with probability $\nu(x)$.
In time
steps $t, t-1,\ldots, 1$, the random walk proceeds as prescribed
by $S$.
At time step $0$, each state $x\in\Omega$ transitions
to $r$ with probability $\mu(x)$
and to an arbitrary state in
$\Omega$ with probability $1-\mu(x)$.
Similarly, $B^t$ denotes
the joint random variable 
$B^t = (B^t_{-1},B^t_0,\ldots,B^t_{t+1})$,
where $B^t_i$ is the
location of the walk at time step $i$.

The following facts about $F^t$ and $B^t$ are immediate.
The random variables $F^t_{-1}$ and $B^t_{t+1}$ are fixed
to a single value $r$. The random variables $F^t$, $B^t$ are Markovian,
namely,
$\dist(F^t_i \emid F^t_{i-1},\ldots,F^t_{-1}) 
    = \dist(F^t_i \emid F^t_{i-1})$ and 
$\dist(B^t_{i-1} \emid B^t_{i},\ldots B^t_{t+1})
    = \dist(B^t_{i-1} \emid B^t_{i})$ for 
$i\in \set{0,\ldots, t+1}$.

\subsection{Random walks returning to the origin}
Assume that $\Pr[F^t_{t+1}=r]>0$.
Let $X$ be the walk $F^t$ conditioned on
$F^t_{t+1} = r$.
Note that $X$ is a random variable on the sample space
$\OC^{t+3}$. The next two lemmas explicitly calculate 
the distribution of $X$.

For a matrix $M\colon\Omega\times \Omega\to\realspos$,
functions $f,g\colon\Omega\to\realspos$, and $x,y\in \Omega$
we use the shorthands
\begin{align*}
M(f,y)&\defeq \sum_{x\in\Omega} f(x)M(x,y) = (M^\transpose f)(y)\\
M(x,g)&\defeq \sum_{y\in\Omega} M(x,y)g(y) = (Mg)(x)\\
M(f,g)&\defeq \sum_{x,y\in\Omega} f(x)M(x,y)g(y) 
    = f^\transpose M g,
\end{align*}
where the last expression in each line is understood as
a matrix vector multiplication.
\begin{lemma}
\label{lem:xdist}
Under our assumption $S^t(\mu, \nu) > 0$,
\begin{enumerate}[(i)]
\item we have
  $\Pr\sparen{X_i = x} 
      = \frac{S^i(\mu,x)S^{t-i}(x,\nu)}{S^t(\mu,\nu)}$, and
\item if $S^{t-i}(x,\nu)>0$, 
    we have $\Pr[X_{i+1}=y \emid X_{\le i} = x_{\le i}]
    =\frac{S(x_i,y)S^{t-i-1}(y,\nu)}{S^{t-i}(x_i,\nu)}$.
\end{enumerate}
\end{lemma}
\begin{proof}
From the definition of $F^t$ (cf.\ \autoref{sec:refwalk}), 
we have 
  \begin{align}
    \Pr[F^t_i =x] &= S^i(\mu, x)\label{teq:xda}\\
    \Pr[F^t_{t+1}=r\emid F^t_{i}=x] 
        &= S^{t-i}(x,\nu)\label{teq:xdb}\\
    \Pr[F^t_{t+1}=r] &= S^t(\mu,\nu)\label{teq:xdc}\\
    \Pr[F^{t}_{i+1}=y\text{ and }F^t_{t+1}=r\emid 
        F^{t}_i =  x] &= S(x,y)S^{t-i-1}(y,\nu)\label{teq:xde}.
  \end{align}
%TODO(saglam): some more elegant citing mechanism
Using Bayes' rule with \autoref{teq:xda},
\eqref{teq:xdb} and \eqref{teq:xdc} gives (i).
Combining \autoref{teq:xdb}, \eqref{teq:xde}
and the observation that $F^t$ is Markovian
gives (ii).
\end{proof}

With \autoref{lem:xdist} we confirm that the random variable $X=(X_{-1},X_0,\ldots,X_{t+1})$ is Markovian; 
in particular a time inhomogeneous random walk on $\OC$.
Next we observe that the random variable $B^t$ conditioned
on $B^t_{-1}=r$ is precisely $X$ also.

\begin{lemma}
\label{lem:xbackwards}
Under our assumption $S^t(\mu,\nu)>0$,
\begin{enumerate}[(i)]
  \item we have $\dist(X) = \dist(B^t\emid B^t_{-1}=r)$, and
  \item if $S^i(\mu, x)>0$, we have 
      $\Pr[X_{i-1}=y\emid X_{\ge i} = x_{\ge i}] = 
          \frac{S(x_i,y)S^{i-1}(y,\mu)}{S^i(x_i,\mu)}$.
\end{enumerate}
\end{lemma}
\begin{proof}
For any $x\in\OC^{t+3}$ with $x_{t+1} = r$,
  \begin{align*}
    \Pr[X = x] &= \frac{\mu(x_0)\prod_{i=1}^t S(x_{i-1},x_i)\nu(x_t)}
    {S^t(\mu,\nu)}\\
    &=\frac{\nu(x_t)\prod_{i=1}^t S(x_i, x_{i-1}) \mu(x_0)}
        {S^t(\mu,\nu)}
            &&\text{(as $S$ is symmetric)}\\
    &= \frac{\Pr[B^t = x]}{\Pr[B^t_{t+1} = r]}
     = \Pr[B^t = x \emid B^t_{t+1} = r]
        &&\text{(by Bayes' rule).}
  \end{align*}
This proves (i). Given (i), the proof of (ii) is 
the same as \autoref{lem:xdist}(ii).
\end{proof}

\begin{lemma}
\label{lem:xvsf}
We have $\kldiv{X}{F^t} = \kldiv{X}{B^t} = -\log S^t(\mu,\nu)$.
\end{lemma}
\begin{proof}
Recall that $\Pr[F^t_{t+1}=r] = S^t(\mu,\nu)$. 
Since $X$ is obtained from $F^t$ by
conditioning on $F^t_{t+1}=r$, the equality
criteria of \autoref{lem:klcond} are
fulfilled and thus $\kldiv{X}{F^t}=-\log S^t(\mu,\nu)$.
The derivation of $\kldiv{X}{B^t}$ is
identical as per \autoref{lem:xbackwards}(i).
\end{proof}

\subsection{Longer random walks}
\label{sec:zdef}
Let $J$ be an integer valued random variable taking
the values $\set{1,2,\ldots,t}$, each with equal probability. For each fixing
$j$ of $J$ we perform a random walk $Z\emid J=j$ on 
$\OC$ that evolves
in time steps $-1,0,1,\ldots,t,t+1,t+2,t+3$ as follows.

The random walk starts at $r$ and for each time step 
$-1\le i <j$, proceeds according to the
transition kernel $\dist(X_{i+1}\emid X_{i})$. At time step
$j$, the random walk proceeds according to 
$\dist(X_{j-1}\emid X_{j})$ and in time steps 
$j< i \le t+3$ proceeds according to the transition kernel 
$\dist(X_{i-1}\emid X_{i-2})$. We view $Z$ as a joint random
variable $Z=(Z_{-1},Z_0,\ldots,Z_{t+3})$, where $Z_i$ denotes
the location of the random walk at time step $i$.

\begin{lemma}
\label{lem:ydist}
For $-1\le i\le j$, we have 
$\dist(Z_i\emid J=j) = \dist(X_i)$ and for 
$j< i \le t + 3$,
$\dist(Z_i\emid J=j) = \dist(X_{i-2})$.
\end{lemma}
\begin{proof}
This follows from the fact that 
\begin{align*}
\dist(F^t \emid F^t_{t+1} = r) = \dist(X)
    = \dist(B^t \emid B^t_{-1} = r)
\end{align*}
and that $X$ is an actual random walk (i.e., Markovian)
on $\OC$.

To be more explicit, we have $\dist(X_i)=\dist(Z_i)$ for 
$i\le j$ since both $X$ and $Z$ start at $r$ in time step 
$-1$ and evolve according to the transition kernel
$\dist(X_{i+1}\emid X_i)$ for $i=-1,\ldots, j-1$. Since
$\dist(X_j) = \dist(Z_j)$ and $Z$ proceeds according to
$\dist(X_{i-1}\emid X_j)$ at time step $j$, by 
\autoref{lem:xbackwards}, $\dist(X_{j-1}) = \dist(Z_{j+1})$.
Finally in time steps $i>j$, we have $\dist(Z_i)=\dist(X_{i-2})$
since $\dist(X_{j-1}) = \dist(Z_{j+1})$ and $Z$ proceeds according
to $\dist(X_{i-1}\emid X_{i-2})$.
\end{proof}

From this we can deduce that $Z$ always
ends up in $r$ at time step $t+3$. We next
argue that if $X$ does not diverge too much from
the reference random walk $F^t$, then $Z$ does 
not diverge too much from $F^{t+2}$.

\begin{lemma}
\label{lem:zdiv}
We have
\begin{align*}
\kldiv{Z\emid J}{F^{t+2}} = \frac{t+2}{t}\kldiv{X}{F^t} 
  &- \frac{1}{t}\nparen{
    \kldiv{X_0}{F^t_0} + 
    \kldiv{X_{t+1}\emid X_{t}}{F^t_{t+1}\emid F^t_t}
  }\\
  &- \frac{1}{t}\nparen{
    \kldiv{X_t}{B^t_t} + 
    \kldiv{X_{-1}\emid X_{0}}{B^t_{-1}\emid B^t_0}
  }.
\end{align*}
\end{lemma}
\begin{proof}
For a fixing $j$ of $J$, we have 
\begin{align*}
\kldiv{Z\emid J=j}{F^{t+2}} 
    &= \sum_{i=-1}^{j-1}\kldiv{X_{i+1} \emid X_{i}}
        {F^{t+2}_{i+1} \emid F^{t+2}_{i}}
    + \kldiv{X_{j-1}\emid X_j}{F^{t+2}_{j+1}
                    \emid F^{t+2}_{j}}\\
    &\quad\quad
        + \sum_{i=j+1}^{t+2}\kldiv{X_{i-1}\emid X_{i-2}}
      {F^{t+2}_{i+1}\emid F^{t+2}_{i}},
\end{align*}
where we have used the chain rule for divergence (cf.\ 
\autoref{lem:klchain}), the fact that $Z\emid J=j$ and 
$F^{t+2}$ are Markovian and \autoref{lem:ydist}. 
Recalling that $F^t$ and $B^t$ evolve according 
to $S$ in time steps $0,1,\ldots,t-1$, and 
$\dist(F^t_{t+1}\emid F^t_t)
    =\dist(F^{t+2}_{t+3}\emid F^{t+2}_{t+2})$,
we write
\begin{align*}
\kldiv{Z\emid J=j}{F^{t+2}}
  &= \sum_{i=-1}^t \kldiv{X_{i+1}\emid X_i}{F^t_{i+1}\emid F^t_i}\\
     &\qquad\qquad
     + \kldiv{X_{j-1}\emid X_j}{B^t_{j-1}\emid B^t_{j}}
     + \kldiv{X_j\emid X_{j-1}}{F^t_{j}\emid F^t_{j-1}}\\
  &= \kldiv{X}{F^t} 
    + \kldiv{X_{j-1}\emid X_j}{B^t_{j-1}\emid B^t_{j}}
    + \kldiv{X_j\emid X_{j-1}}{F^t_{j}\emid F^t_{j-1}}
\end{align*}
again by the chain rule for divergence 
(\autoref{lem:klchain}) and the
fact that $X$ and $F^t$ are Markovian. Now taking 
an expectation over all $j\in \supp(J)$, we have
\begin{align*}
\kldiv{Z\emid J}{F^{t+2}} 
  &= \frac{1}{t}\sum_j\kldiv{Z\emid J=j}{F^{t+2}}\\
  &= \kldiv{X}{F^t} 
    + \frac{1}{t}\sum\nparen{\kldiv{X_{j-1}\emid X_j}
                            {B^t_{j-1}\emid B^t_{j}}
    + \kldiv{X_j\emid X_{j-1}}{F^t_{j}\emid F^t_{j-1}}}\\
  &= \kldiv{X}{F^t}
  + \frac{\kldiv{X}{B^t} 
      - \kldiv{X_t}{B^t_t}
      - \kldiv{X_{-1}\emid X_0}{B^t_{-1}\emid B^t_0}
    }{t}\\
  &\qquad\qquad\qquad\;+ \frac{\kldiv{X}{F^t} 
      - \kldiv{X_0}{F^t_0}
      - \kldiv{X_{t+1}\emid X_t}{F^t_{t+1}\emid F^t_t}
    }{t}.
\end{align*}
Since $\kldiv{X}{B^t} = \kldiv{X}{F^t}$ by 
\autoref{lem:xvsf}, collecting the $\kldiv{X}{F^t}$
terms we finish the proof.
\end{proof}

Finally, we lower bound the negative terms in the statement
of \autoref{lem:zdiv}.
\begin{lemma}
\label{lem:negterms}
We have
\begin{align}
\kldiv{X_0}{F^t_0}
    + \kldiv{X_{-1}\emid X_{0}}{B^t_{-1}\emid B^t_0}
    &\ge \ryent{2}{\mu} \defeq -\log \lnorm{\mu}{2}^2,
        \text{ and} \label{teq:negterm1}\\
\kldiv{X_t}{B^t_t}
    + \kldiv{X_{t+1}\emid X_{t}}{F^t_{t+1}\emid F^t_t}
    &\ge \ryent{2}{\nu} \defeq -\log \lnorm{\nu}{2}^2,
        \label{teq:negterm2}
\end{align}
where $\ryent{2}{\cdot}$ denotes the 
second order Rényi entropy. 
%If the inequalities hold
%with equality simultaneously, then $\mu = \nu$ and 
%$\mu$ is an eigenvector of $S$.
\end{lemma}
\begin{proof}
We only prove the first inequality as the
second one is symmetric. By \autoref{lem:xdist},
we have
\begin{align*}
\kldiv{X_0}{F^t_0} &= 
  \sum_{x\in \Omega}
      \frac{\mu(x)S^{t}(x,\nu)}{S^t(\mu,\nu)}
      \log \frac{S^{t}(x,\nu)}{S^t(\mu,\nu)}
          \quad\text{ and}\\
\kldiv{X_{-1}\emid X_{0}}{B^t_{-1}\emid B^t_0} &=
  \sum_{x\in \Omega}
      \frac{\mu(x)S^{t}(x,\nu)}{S^t(\mu,\nu)}
      \log\frac{1}{\mu(x)}.
\end{align*}
Let $\Psi=\supp(X_0)$. By adding the two terms we get
\begin{align}
\kldiv{X_0}{F^t_0}+
\kldiv{X_{-1}\emid X_{0}}{B^t_{-1}\emid B^t_0} &=
  -\sum_{x\in\Psi}
      \frac{\mu(x)S^{t}(x,\nu)}{S^t(\mu,\nu)}
      \log\frac{\mu(x)S^t(\mu,\nu)}{S^{t}(x,\nu)}
          \nonumber\\
  &\ge -\log \sum_{x\in\Psi}\frac%
    {\mu(x)^2 S^{t}(x,\nu)S^t(\mu,\nu)}
      {S^t(\mu,\nu)S^{t}(x,\nu)}
          \label{eq:norm-u-jens}\\
  &=   -\log \sum_{x\in\Psi} \mu(x)^2
      \nonumber\\
  &\ge -\log \sum_{x\in\Omega} \mu(x)^2\;,
      \label{eq:psivsomega}
\end{align}
where the first inequality is by concavity of 
$z\mapsto\log z$
and the second inequality is true as the 
summands are nonnegative.
\end{proof}

\subsection{Combining the inequalities}
\label{sec:bd-final}
\begin{proof}[Proof of \autoref{thm:blakley-dixon}]
Note that $Z_{-1}$ is fixed to $r$ by definition 
(cf.\ \autoref{sec:zdef}) and $Z_{t+3}$ is fixed to $r$
by \autoref{lem:ydist}. Therefore by 
\autoref{lem:klcond} we have
\begin{align}
-\log S^{t+2}(\mu,\nu)
  &\le \kldiv{Z}{F^{t+2}} \label{eq:kvk2} \\
  &= \kldiv{Z\emid J}{F^{t+2}} - \muti{J}{Z}
      \label{eq:ztozcondj}\\
  &\le \kldiv{Z\emid J}{F^{t+2}}
      \label{eq:mutrelax}\\
  &\le \frac{t+2}{t}\kldiv{X}{F^t} 
      + \frac{\log\lnorm{\mu}{2}^2 + \log\lnorm{\nu}{2}^2}
            {t}.
        \nonumber
\intertext{%
Here \autoref{eq:ztozcondj} follows from the chain rule 
for the divergence (\autoref{lem:klchain}) 
and the definition of mutual information 
(cf.\ \autoref{eq:muti-def}),
\autoref{eq:mutrelax} follows from
the nonnegativity of mutual information and
the last line follows from 
\autoref{lem:zdiv} and \autoref{lem:negterms}.
Plugging in $\kldiv{X}{F^t}=-\log S^t(\mu,\nu)$, provided by 
\autoref{lem:xvsf}, we obtain}
  -\log S^{t+2}(\mu,\nu)
    &\le -\frac{t+2}{t}\log S^t(\mu,\nu) 
         + \frac{\log\lnorm{\mu}{2}^2 + \log\lnorm{\nu}{2}^2}
            {t}.
        \nonumber
\end{align}
Arranging, we get
\begin{align*}
\lnorm{\mu}{2}^2 \lnorm{\nu}{2}^2\iprod{\nu}{S^{t+2}\mu}^t
  \ge
    \iprod{\nu}{S^t\mu}^{t+2}
\end{align*}
and substituting $\mu = u/\lone{u}$, 
$\nu = v/\lone{v}$, and recalling that $u,v$ 
are unit vectors, we obtain
\begin{align}
\iprod{v}{S^{t+2}u}^t 
&\ge \iprod{v}{S^t u}^{t+2}, \text{ i.e.,} \nonumber\\
m_{t+2} &\ge m_t^{1+2/t}\label{eq:absclaim1}.
\end{align}
By applying this inequality iteratively, we get
$\iprod{v}{S^{k}u}^t\ge\iprod{v}{S^t u}^{k}$ or
written differently
$m_k^{1/k}\ge m_t^{1/t}$
as long as $k>t$ and $k,t$ have the same parity.

Next we characterize the equality conditions of 
\autoref{eq:absclaim1}. Let us verify the
`if' direction of the statement. 
Clearly if $\iprod{v}{S^ku}=0$ then we have $\iprod{v}{S^ku} = \iprod{v}{S^tu}$ by the
first part of the theorem and the fact that 
$\iprod{v}{S^tu}\ge 0$. If $S^2u = \lambda^2 u$, 
then $m_{2t} = \lambda^{2t}$ and if $Su=\lambda v$ and $Sv=\lambda u$, then $m_{2t+1}=\lambda^{2t+1}$, therefore
in both $t$ even and $t$ odd cases 
the inequality holds with equality.

Conversely, if 
$0\neq \iprod{v}{S^{k+2}u}=\iprod{v}{S^ku}$, then the
inequalities
\eqref{eq:norm-u-jens}, 
\eqref{eq:psivsomega},
\eqref{eq:kvk2},
 and
\eqref{eq:ztozcondj} must hold with equality.
Combining the assumption that \autoref{eq:psivsomega}
and \eqref{eq:norm-u-jens} 
hold with equality with the strict concavity of
$z\mapsto \log z$ and Jensen's lemma, we get that 
$S^t\nu = \lambda_1\mu + \sigma_1$ for some $\lambda_1\le 1$ and $\sigma_1\in\reals^{\Omega}_{+}$ satisfying 
$\supp(\sigma_1)\cap\supp(\mu)=\emptyset$. This also means that $\Pr[X_0 = x] = \mu(x)^2/\lnorm{\mu}{2}^2$ for $x\in \Omega$ and a similar and symmetrical argument shows that
$\Pr[X_t = x] = \nu(x)^2/\lnorm{\nu}{2}^2$ for 
$x\in\Omega$.
Assuming \autoref{eq:ztozcondj} holds 
with equality leads to $\muti{Z}{J}=0$, which in turn
shows that $\dist(X_i)=\dist(X_{i+2})$ for 
$i=0,\ldots,t-2$. 
Let 
$X^{k+2} \defeq \nparen{F^{t+2}\emid F^{t+2}_{t+3}=r}$.
From our assumption $0\neq \iprod{v}{S^{k+2}u}=\iprod{v}{S^ku}$ we conclude that $\dist(Z)=\dist(X^{t+2})$
as $X^{t+2}$ is the minimally divergent distribution from $F^{t+2}$ among distributions on walks ending at state $r$. Since 
$\dist(X^{t+2}_2)=\dist(X^{t+2}_0)=\dist(X_0)$, and $S^t\nu=\lambda_1\mu+\sigma_1$ we get that 
$S^2\mu = \lambda_2 \mu + \sigma_2$  for some $\lambda_2\le 1$ and $\sigma_2\in\reals^{\Omega}_{+}$ satisfying 
$\supp(\sigma_2)\cap\supp(\mu)=\emptyset$. Now we will show that it must be that $\sigma_2=0$. Suppose for sake of contradiction that $\sigma_2(z)>0$ for some $z\notin \supp(\mu)$. There exists $x,y\in\Omega$ so that $\mu(x)S(x,y)S(y,z)>0$. If $y\in\supp(X_1)$ then adding to a walk $w\in \supp(X)$ with $w_1=y$ the loop $(y,z)(z,y)$
we obtain a length $t+2$ walk which is not in the support of $X^{t+2}$ as $\dist(X^{t+2}_2)=\dist(X_2)=\dist(X_0)$,
which contradicts the fact that 
$X^{t+2}$ is defined as $F^{t+2}\emid F^{t+2}_{t+3}=r$.
If on the other hand $y\notin\supp(X_1)$, 
adding the loop $(x,y)(y,x)$ to a walk with $w_0=x$ 
leads to a walk which is not in the support of $X^{t+2}$, 
which is a contradiction. Having established $S^2\mu = \lambda \mu$, we complete the proof for even $t$
by recalling that $\dist(X_0)=\dist(X_t)$ therefore $\mu=\nu$. For $t$ odd, the last argument shows that
 $S\mu = \lambda_3\nu + \sigma_3$ for some $\lambda_3\le 1$ and $\sigma_3\in\reals^{\Omega}_{+}$ satisfying 
$\supp(\sigma_3)\cap\supp(\mu)=\emptyset$. 
It remains to show that $\sigma_3=0$ 
by using the assumption that $\dist(Z)=\dist(X^{t+2})$. Suppose $\sigma_3(y)>0$ for some $y\notin\supp(\nu)$. There exists $x\in\Omega$ such that $\mu(x)S(x,y)>0$. Adding the loop $(x,y)(y,x)$ to
a walk $w$ with $w_{t-1}=x$ leads to a length $t+2$ walk
which is not in the support of $X^{t+2}$. A symmetrical
argument shows that $\lambda'\mu = S\nu$. This completes the proof.
\end{proof}
% !TeX root = heatdiscrete.tex
\section{Near log-convexity of 
         $t\mapsto m_{2t}$ and $t\mapsto m_{2t+1}$}
\label{sec:logconvex}
In this section we would like to prove the following
improvement to \autoref{eq:absclaim1}: for all $\eps>0$ 
there exists a $\delta>0$ such that
\begin{align}
m_{t+2}
\ge m_{t}^{1+2/t}\cdot
\min\set{t^{1-\eps}, \ceil{\delta\frac{m_t^{1-2/t}}{m_{t-2}}}},
\quad \forall t \ge 2.
\label{eq:logconv}
\end{align}
Recall that in proving \autoref{eq:absclaim1}, 
in line \ref{eq:mutrelax}, we used the relaxation 
$\muti{J}{Z}\ge 0$.
Note that $J$ is uniformly distributed on $[t]$
therefore has $\log t$ bits of entropy and provided
that it is possible to infer $J$ from $Z$ 
(i.e., it is possibly to locate the time reversal we 
have inserted in $Z$)
the $\muti{J}{K}$ term appears to be large enough to 
recover the factor
\begin{align*}
\min\set{t^{1-\eps}, 
  \ceil{\delta\frac{m_t^{1-2/t}}{m_{t-2}}}}.
\end{align*}
Note moreover that intuitively we are able to 
infer $J$ from $Z$ better when 
$\frac{m_t^{1-2/t}}{m_{t-2}}$ is high, 
as in such cases on average for a time step $i\in[t]$ 
and a typical $x\sim X_i$,
the distributions $\dist(X_{i-1}\emid X_i=x)$ and
$\dist(X_{i+1}\emid X_i=x)$ should be far from each other, 
as otherwise we can argue that there should be many $t-2$ 
walks as follows. If $\dist(X_{i-1}\emid X_i=x)$ and
$\dist(X_{i+1}\emid X_i=x)$ are close to each other, 
there should be many $p\in\Omega$ which has high probability
in both
these distributions. Sample such a $p$, and 
attach to it a walk sampled from $X_{-1}X_1\ldots X_{i-1}\emid X_{i-1}=p$ 
and another walk sampled from $X_{i+1}X_{i+2}\ldots X_{t+1}\emid X_{i+1}=p$,
which leads to a length $t-2$ walk returning to the origin.
However if $m_{t-2}$ is low, this should not happen and therefore
$\dist(X_{i-1}\emid X_i=x)$ and
$\dist(X_{i+1}\emid X_i=x)$ on average should be far apart, 
which means that we can notice when we take a step backwards in time
and therefore infer $J$. In particular, \autoref{fig:ex1} 
gives such an example where $m_{t-2}=0$ and we can always
recover $J$ with certainty from a sample from $Z$: whenever we take
a step to the left, it must be that we are at time step $J$.

Given this discussion, a direct approach to proving 
\autoref{eq:logconv} appears to bound 
\begin{align}
\muti{Z}{J}\ge \log \min\set{t^{1-\eps}, 
\ceil{\delta\frac{m_t^{1-2/t}}{m_{t-2}}}}.
\label{eq:hope}
\end{align}
Unfortunately, this approach does not seem to 
work as we demonstrate with an example in the full version of this paper.
The problem here appears to be that we 
fix a single distribution $Z$ to explore the
two cases of \autoref{eq:logconv}. 
In our final approach, we pick different 
distributions depending on the case we would like 
to prove. Namely, if 
$\muti{J}{Z}\ge (1-\eps) \log t$, then carrying 
out the calculations in \autoref{eq:kvk2} through 
\ref{eq:absclaim1} with the assumption 
$\muti{J}{Z}\ge (1-\eps) \log t$, we prove 
the first case, namely $m_{t+2}\ge t^{1-\eps}m_t^{1+2/t}$ 
using the distribution $Z$.
If $\muti{J}{Z}< (1-\eps) \log t$ on the other hand, 
we demonstrate two new random variables $W,Y$ which are 
distributed respectively on length $t+2$ and length $t-2$ 
paths so that 
  $\kldiv{W}{F^{t+2}} + \kldiv{Y}{F^{t-2}}\le 
    -2\log S^t(\mu, \nu) - \log \delta$,
which implies that $m_{t-2}m_{t+2}\ge \delta m_t^2$.
While $W$ and $Y$ are constructed by modifying $X$ in 
suitable ways, which is how $Z$ was constructed also,
we do so with the hindsight of 
having inspected what causes $\muti{J}{Z}$ to be smaller
than $(1-\eps)\log t$. It is precisely this adaptivity
which enables this approach to overcome the difficulties
encountered by the one suggested in 
\autoref{eq:hope}.

If $\muti{J}{Z}\ge (1-\eps) \log t$, by plugging this into 
\autoref{eq:mutrelax} and carrying out the 
following calculations, 
we get $m_{t+2} \ge t^{1-\eps}\cdot m_t^{1+2/t}$.
Therefore it remains to show there exists a $\delta>0$
such that assuming $\muti{J}{Z} < (1-\eps) \log t$,
we have
$m_{t+2}m_{t-2}\ge \delta m_t^2$.
To do so, we will demonstrate distributions 
$W$ and $Y$
on walks that start from $r\in \OC$
and return to $r$ 
after spending respectively $t+2$ and $t-2$ time 
steps in 
$\Omega$ such that $\kldiv{W}{F^{t+2}} + 
\kldiv{Y}{F^{t-2}}\le -2\log S^t(\mu, \nu) - 
\log \delta$. 
Notice that by \autoref{lem:jensen} this indeed implies
that $m_{t+2}m_{t-2}\ge \delta m_t^2$.
The distributions $W$ and $Y$ will be mixture of $\Theta(t)$
random walks, in particular, they are not Markovian in 
general.

For brevity let us set $\mu_i^x \defeq \dist(X_{i-1}\emid X_i = x)$ and 
$\nu_i^x \defeq \dist(X_{i+1}\emid X_i=x)$.
Let $U$ be the unary encoding of $J$: a length $t$ bit
vector of which only the $J$th coordinate is set.
First we would like to understand 
the contribution of each bit of $U$ to 
$\muti{Z}{J}=\muti{Z}{U}$.
Using the chain rule, we write
\begin{align}
(1-\eps)\log t&>   \muti{U}{Z} \nonumber\\
              &=   \sum_{i=1}^t \muti{U_i}{Z\emid U_{<i}}
                    \label{eq:t1-chain}\\
              &=   \sum_{i=1}^t \frac{t-i+1}{t}
                   \muti{U_i}{Z\emid U_{<i}=0}
                   \label{eq:t1-condition}\\
           &\ge \sum_{i=1}^t \frac{t-i+1}{t}
                \muti{U_i}{Z_iZ_{i+1}\emid U_{<i}=0}
                \label{eq:t1-data}\\
           &=   \sum_{i=1}^t\frac{1}{t}
                \E_{x\sim X_i}\kldiv{\mu_i^x}
                {\lambda_i \mu_i^x + (1-\lambda_i)\nu^x_i}
                \nonumber\\
           &\quad\quad\quad\quad\quad
                + \sum_{i=1}^t\frac{t-i}{t}
                \E_{x\sim X_i}\kldiv{\nu_i^x} 
                {\lambda_i\mu_i^x + (1-\lambda_i)\nu^x_i}
                \label{eq:t1-split}
\end{align}
where we set $\lambda_i \defeq 1/(t-i+1)$, which is
the probability that $U_i=1\emid U_{<i}=0$.
Here, \autoref{eq:t1-chain} follows from the 
chain rule, \autoref{eq:t1-condition} is 
true because if $U_{<i}\neq 0$ then $U_i=0$ 
(as $U$ has a single coordinate that is one) 
and consequently the mutual information is
zero, and \autoref{eq:t1-data} is the data 
processing inequality. Next we lower bound 
\autoref{eq:t1-split} by its first term 
(which is valid since $\mu^x_i, \nu^x_i$ 
are distributions hence the second term of 
\autoref{eq:t1-split} is nonnegative), obtaining
\begin{align}
       (1-\eps) \log t &> \E_{i\sim J}\E_{x\sim X_i}
                \kldiv{\mu_i^x}
                      {\lambda_i \mu_i^x 
                      + (1-\lambda_i)\nu^x_i}\,.
                \label{eq:mut-inf}
\end{align}

To simplify the presentation, here we only 
provide the proof of \autoref{thm:main} for 
$\eps >7/8$ which demonstrates the ideas in 
their simplest form. This bound already implies all 
our results in complexity theory, with a constant
factor loss of no more than 8.
The proof for any $\eps>0$ can be found in the 
full version of this paper.

\subsection{The bound for $\eps> 7/8$}
If we condition on the event $i\in\set{1,\ldots,\lceil t/2 \rceil}$, this expectation
increases by a factor of at most 2; namely
\begin{align*}
\E_{i\sim \sparen{t/2}}
  \E_{x\sim X_i}
      \kldiv{\mu_i^x}
      {\lambda_i \mu_i^x 
          + (1-\lambda_i)\nu^x_i} < 2(1-\eps)\log t.
\end{align*}
By Markov's inequality
\begin{align*}
\Pr_{i\sim\sparen{t/2},x\sim X_i}
  \sparen{\kldiv{\mu_i^x}
      {\lambda_i \mu_i^x 
          + (1-\lambda_i)\nu^x_i} \ge 8(1-\eps)\log t} < 1/4,
\end{align*}
so it follows that there is a set 
$T\subseteq\sparen{\lceil t/2\rceil}$ of size
at least $\floor{t/4}$ such that if $i\in T$ we have
\begin{align}
\Pr_{x\sim X_i}
  \sparen{\kldiv{\mu_i^x}
      {\lambda_i \mu_i^x 
          + (1-\lambda_i)\nu_i^x} \ge 8(1-\eps)\log t} < 1/2.
          \label{eq:halfp}
\end{align}
For each $i\in T$ let $X'_i$ be the random variable 
obtained from $X_i$ by conditioning on those 
$x\in \supp(X_i)$ satisfying $\kldiv{\mu_i^x}
{\lambda_i \mu_i^x + (1-\lambda_i)\nu^x_i}<8(1-\eps)\log t$.
Furthermore, for each $i\in T$ and $x\in\supp(X'_i)$,
we construct distributions
$\pi_i^x\colon \Omega\to\realspos$ to be specified later.
Let $P_i$ be sampled by $x\sim X'_i$ first and then picking
$p\sim \pi_i^x$. 

\subsection{The distributions $W$ and $Y$}
Let $K$ be an integer sampled uniformly at
random from the set $T$ (constructed in 
the previous section).
For each fixing $k$ of $K$, the random variables $W\emid K=k$
and $Y\emid K=k$ are random walks 
(i.e., they are Markovian) 
constructed as follows. We first pick $x,p\sim X'_kP_k$.
The walk $Y\emid K=k$ is generated by concatenating a sample
from $X_{-1}X_0\ldots X_{k-1}\emid X_{k-1} = p$ and 
an independent sample from $X_{k+1}\ldots X_{t+1}\emid X_{k+1} = p$.
The walk $W$ is generated by concatenating a sample from 
$X_{-1}X_0\ldots X_{k}\emid X_{k} = x$, the path
$(x, p)$ and $(p,x')$ for an independent sample 
$x'\sim\nparen{X'_k\emid P_k = p}$ and 
an independent sample from $X_{k}\ldots X_{t+1}\emid X_{k} = x'$.

For $k\in T$ we define another random walk 
$\check{X}^k = 
(\check{X}^k_{-1},\ldots,\check{X}^k_{t+1})$, 
only to be used in the analysis of
$W$ and $Y$. We sample $x\sim X'_k$
and set $\check{X}^k_k = x$. We pick the rest of the coordinates of $\check{X}^k$ according to the distribution $X\emid X_k = x$. Note that for any $k\in T$, we have
\begin{align*}
\kldivs{\check{X}^k}{X} = \kldiv{X'_k}{X_k}\le 1
\end{align*} by \autoref{eq:halfp}
and \autoref{lem:klcond} and the fact that both $X$ and 
$\check{X}^k$ are Markovian.

\begin{lemma}
\label{lem:wy}
We have
\begin{align*}
\kldiv{W\emid K=k}{F^{t+2}} +& \kldiv{Y\emid K=k}{F^{t-2}}\\
&\le -2\log S^t(\mu,\nu) + 2 + \E_{x\sim X'_k}\kldiv{\pi_k^x}{\mu_k^x} + 
        \E_{x\sim X'_k}\kldiv{\pi_k^x}{\nu_k^x}.
\end{align*}
\end{lemma}
\begin{proof}
We have
\begin{align*}
\kldiv{W\emid K = k}{F^{t+2}} &= \kldivs{\check{X}^k}{F^t} + 
  \kldiv{P_k\emid X'_k}{F_{k+1}\emid F_{k}}
  +\kldiv{X'_k\emid P_k}{F_{k+1}\emid F_{k}}
\end{align*}
and further
\begin{align*}
\kldiv{Y\emid K = k}{F^{t-2}} +& 
  \kldiv{P_k\emid X'_k}{F_{k+1}\emid F_{k}}
  +\kldiv{X'_k\emid P_k}{F_{k+1}\emid F_{k}}
  \\&= \kldivs{\check{X}^k}{F^t}+\E_{x\sim X'_k}
      \kldiv{\pi_k^x}{\mu_k^x} + 
      \E_{x\sim X'_k}\kldiv{\pi_k^x}{\nu_k^x}.
\end{align*}
Summing up the two inequalities and substituting 
$\kldivs{\check{X}^k}{X}\le 1$ we get the result.
\end{proof}
At this point, in light of \autoref{lem:wy},
we could pick each $\pi_k^x$ so that it minimizes
$\kldiv{\pi_k^x}{\mu_k^x} + \kldiv{\pi_k^x}{\nu_k^x}$:
the unique minimizer is given by 
$\pi_k^x = \sqrt{\mu_k^x\nu_k^x} / 
\iprod{\sqrt{\mu_k^x}}{\sqrt{\nu_k^x}}$. 
However doing so leads to $W,Y$ which 
diverge from the $F$ walk by more than a constant,
and therefore is not good enough for our needs.
To obtain better random variables $W$ and $Y$, 
we crucially use the fact that 
$W$ is a mixture of $\Theta(t)$ random walks. Namely, if
we consider the entropy coming from the $\muti{W}{K}$ term
also, a better strategy for picking the distributions 
$\pi_k^x$ becomes available.
By contrast, we do not use the fact that $Y$ is a mixture
and, in fact, it can be replaced by $Y\emid K=k_0$ where
$k_0 = \arg\min_k\kldiv{Y\emid K = k}{F^{t-2}}$, however
the averaged quantity $\kldiv{Y\emid K}{F^{t-2}}$ is
far more convenient to work with.

\subsection{The contribution of $\muti{K}{W}$}
Similar to \autoref{eq:mut-inf}, we would like to understand
the contribution of each time step $t\in T$ to 
$\muti{K}{W}$. Let $V$ be the unary encoding of $K$:
a length $t$ bit vector of which only the $V$th 
coordinate is set. Using the chain rule for 
mutual information
\begin{align*}
\muti{W}{V} &=   \sum_{i\in T}\muti{V_i}{W\emid V_{<i}}\\
            &\ge \E_{k\sim K} \E_{x\sim X'_k} 
                 \kldiv{\pi_k^x}
           {\eta_i \pi_k^x + (1-\eta_i)\widetilde{v}_k^x},
\end{align*}
where $\eta_k = 1/\rank_{T}(k)$ and 
$\widetilde{v}_k^x 
    \defeq \E_{j > k : j\in T} 
    \dist(\check{X}^j_{k+1}\emid \check{X}^j_k = x)$.
Here $\rank_T(i)$ denotes the position of 
$i\in T$ when the elements of $T$ are sorted in decreasing order.
By \autoref{eq:halfp}, and the definition of
$X'_k$, we have
$\widetilde{v}_k^x(y)\le 2\nu_k^x(y)$ for all $y\in \Omega$.
Therefore we conclude that 
\begin{align}
\muti{W}{K} \ge \E_{k\sim K} \E_{x\sim X'_k} 
                 \kldiv{\pi_k^x}
           {\eta_k \pi_k^x + 2(1-\eta_k)\nu_k^x}.\label{eq:muti-nu}
\end{align}
Note in the above divergence expression the 
reference measure is not a probability distribution, 
which our definition permits (cf.\ \autoref{eq:kldiv}).

Recall our goal in this section is to upper bound 
$\kldiv{W}{F^{t+2}} + \kldiv{Y}{F^{t-2}} 
+ 2\log S^t(\mu, \nu)$ by $\log 1/\delta$. 
Let us write
\begin{align}
\kldiv{W}{F^{t+2}} +& \kldiv{Y}{F^{t-2}} 
+ 2\log S^t(\mu, \nu) \nonumber \\
&\le \kldiv{W\emid K}{F^{t+2}} 
  + \kldiv{Y\emid K}{F^{t-2}} 
  - \muti{K}{W} + 2\log S^t(\mu, \nu)
    \nonumber\\
&\le 2 + \E_{k\sim K, x\sim X'_k}
  \kldiv{\pi^k_x}{\mu^k_x} + 
  \kldiv{\pi^k_x}{\nu^k_x} - \muti{K}{W}
        \nonumber\\
&\le 2 + \E_{k\sim K, x\sim X'_k} \E_{y\sim \pi_k^x} \log
\frac{\eta_k \pi_k^x(y)^2 + 2(1-\eta_k)\nu_k^x(y)\pi_k^x(y)}
{\mu_k^x(y)\nu_k^x(y)}\label{eq:div-min},
\end{align}
where the second inequality follows from 
\autoref{lem:wy} and the last inequality is 
obtained by plugging in \autoref{eq:muti-nu}.
Note that the function $z\mapsto z \log(az^2 + bz)$ is 
strictly convex
in $\realspos$ whenever $ab > 0$, therefore 
for each $k,x$ there is a unique minimizer $(\pi_k^x)^*$
of \autoref{eq:div-min},
which can be calculated, say, using Lagrange multipliers. 
However, instead of the minimizer, we work with a simple approximation of it.
For each $k\in T$ and $x\in \supp(X'_k)$, we let
\begin{align*}
\Psi_k^x \defeq \setbuilder{y\in\Omega}
{\nu_k^x(y)\ge \frac{\lambda_k}{1-\lambda_k}\mu_k^x(y)}.
\end{align*}
By definition of $X'_k$, we have
$\kldiv{\mu_i^x}
      {\lambda_i \mu_i^x 
          + (1-\lambda_i)\nu^x_i} < 8(1-\eps)\log t$.
Let $\gamma = 1- 8(1-\eps)$, which is positive by our assumption
$\eps>7/8$.
By Markov's inequality, and the fact that 
$\lambda_k \le 2/t$,
we get
\begin{align*}
\mu_k^x(\Psi_k^x)\ge \gamma
\end{align*}
for large enough $t$.
Let $\pi_k^x$ be $\mu_k^x \emid \Psi_k^x$, namely we have
$\pi_k^x(y) = \mu_k^x(y) / \mu_k^x(\Psi_k^x)$ if 
$y\in\Psi_k^x$, and $\pi_k^x(y) = 0$ otherwise.
Continuing from \autoref{eq:div-min}, we have
\begin{align}
 &\le 2 +  \E_{k\sim K, x\sim X'_k} \E_{y\sim \pi_k^x} \log
\frac{\eta_k \pi_k^x(y)^2 + 2(1-\eta_k)\nu_k^x(y)\pi_k^x(y)}
{\mu_k^x(y)\nu_k^x(y)}\nonumber\\
 &\le 2+ 
 \E_{k\sim K} \log
\nparen{\frac{\eta_k(1-\lambda_k)}{\lambda_k \gamma^2}+\frac{2}{\gamma}},
\label{eq:setsh}
\end{align}
where the second inequality is true by definition of
$\Psi^x_k$ and $\pi^x_k$. 
Now we argue that the expectation term in \autoref{eq:setsh} is maximized
when $T$ is the set containing the smallest $|T|$ elements of 
$[\ceil{t/2}]$. To see this suppose there is an $i\notin T$
which is smaller than the maximum element of $T$.
Let $j$ be the smallest item in $T$ which is greater than $i$. 
We see that the expectation term
increases if we replace $T$ by $T\setminus\set{j}\cup\set{i}$
as $\log\nparen{\frac{C(1-\lambda_k)}
{\lambda_k \gamma^2}+\frac{2}{\gamma}}$ is decreasing in $k$ 
and the ranks
do not change after swapping $j$ with $i$.
Therefore, 
\begin{align}
\kldiv{W}{F^{t+2}} + &\kldiv{Y}{F^{t-2}} + 2\log S^t(\mu, \nu)
\nonumber \\
&\le 2+
\log\nparen{\prod_{i=1}^{|T|}\frac{t/2 + 3i}{i \gamma^2}}^{1/|T|}
\nonumber \\
&=\log\frac{12}{\gamma^2} + 
\log\nparen{\prod_{i=1}^{|T|}\frac{t/6 + i}{i}}^{1/|T|}
\nonumber \\
&\le \log\frac{12}{\gamma^2} + \log\binom{2|T|}{|T|}^{1/|T|}
\label{eq:binomub}\\
&\le \log\frac{48}{\gamma^2},
\end{align}
where the $\binom{2|T|}{|T|}$ is the middle binomial coefficient,
in the second inequality we use the fact $|T|>t/6$,
and the last inequality is true as $\smash{\binom{2n}{n}< 2^{2n}}$.
Therefore it is enough to choose $\eps > 7/8$ and 
$\delta\le \frac{\nparen{1-8(1-\eps)}^2}{48} = \frac{4}{3}(\eps- \sfrac{7}{8})^2$. We have established
the following.

\begingroup
\def\thetheorem{\ref{thm:main}}
\begin{theorem}[restated]
For any $\eps>7/8$ there is a $\delta>0$ such that
$m_{t+2}\ge t^{1-\eps} m_t^{1+2/t}$ unless 
$m_{t+2}m_{t-2}\ge \delta m_t^2$.
\end{theorem}
\addtocounter{theorem}{-1}
\endgroup

% !TeX root = heatdiscrete.tex
\section{Randomized computational models}
\label{sec:applications}
In this section we show the connection between
\autoref{thm:blakley-dixon}, \autoref{thm:main} and 
the randomized communication and query complexities 
of the Hamming distance problem.

\subsection{Communication complexity}
In a two player communication problem the players, 
named Alice and Bob, receive separate inputs, 
respectively $x\in \mathcal{X}$ and $y\in \mathcal{Y}$, 
and they communicate in order to compute the value 
$f(x,y)$ of a function
$f\colon \mathcal{X}\times \mathcal{Y}\to\binary$ 
(known to both players). 
In an $r$-round protocol, the players can take at most 
$r$ turns alternately sending each other a message 
(that is, a bit string) and the last player to receive 
a message declares the output of the protocol.
A protocol can be {\em deterministic} or {\em randomized}; 
in the latter case the players can base their actions on 
a common random source and we measure the 
{\em error probability}: the maximum over inputs 
$(x,y)\in \mathcal{X}\times\mathcal{Y}$, of the probability 
that the output of the protocol differs from $f(x,y)$. 
The {\em communication cost} of a protocol is the maximum, 
over the inputs and the random string, of the total number 
of bits sent between the players.
For a function 
$f\colon \mathcal{X}\times\mathcal{Y}\to\binary$, 
an integer $r$ and $\delta\in[0,1]$, we denote by 
$R^r_{\delta}(f)$ the minimum over all protocols
for $f$ having $r$-rounds and error probability at most 
$\delta$, of the communication cost incurred. We define 
$R_{\delta}(f)$ similarly, but we take the maximum over 
$\delta$-error protocols with no restriction on the number 
of rounds it uses.

In the $k$-Hamming distance problem, denoted $\Ham^n_k$,
the players receive length-$n$ bit strings, respectively 
$x,y\in\cube$, and are required determine if 
$\lone{x-y}\le k$ or not.
There is a well known one-round communication protocol
which accomplishes this with error probability $\delta$ 
by communicating $O(k\log\nparen{k/\delta})$ bits.

\begin{theorem}
[e.g., Huang, Shi, Zhang and Zhu \cite{HuangSZZ2006}]
\label{thm:ub}
It holds that $$R^1_\delta(\Ham^n_k) = 
O(\min\set{k\log\nparen{k/\delta}, k\log (n/k)}).$$
\end{theorem}

Highly related to the $\Ham^n_k$ is the $k$-disjointness 
problem $\Disj^n_k$, wherein the players each receive a 
$k$-subset of $[n]$ and their goal is to determine if their 
sets intersect.
Notice that $\Disj^n_k$ can be seen as a promise version of
$\Ham^n_{2k-2}$ where each player is guaranteed to have a string
with Hamming weight $k$: the sets are disjoint if and only if the
Hamming distance between the characteristic vectors of the sets
is more than $2k-2$. Therefore any upper bound for the $\Ham^n_k$ 
carries over to $\Disj^n_k$ and any lower bound for $\Disj^n_k$
carries over to $\Ham^n_k$.
Around 1993, 
Håstad and Wigderson \cite{HastadW2007} showed that
there is a more efficient protocol for $\Disj^n_k$ than that implied 
by \autoref{thm:ub}, which communicates only $O(k)$ bits, but over
$O(\log k)$ rounds.

On the lower bounds side, the result of \cite{KalyanasundaramS1992}
implies that $\Omega(k)$ bits is needed for these problems even
if one uses arbitrarily large number of round protocols.
In \cite{BuhrmanGMW2012} it was shown that
any 1-round protocol for $\Disj^n_k$ needs to communicate
at least $\Omega(k \log k)$ bits when $k^2<n$
(this result was proven later in \cite{DasguptaKS12} also).
In Theorem~3.2 of \cite{Saglam2011}, an $\Omega(k\log(1/\delta))$
bound for $1$-round complexity of $\Disj^n_k$ was shown even 
when Bob receives just one element (i.e., the indexing problem) 
for $k<\delta n$ and a 
slightly more general result was shown in \cite{JayramW2011}.
Finally in \cite{SaglamT2013} the communication complexity
of $\Disj^n_k$ was settled:
\begin{align*}
R^r_{1/3}(\Disj^n_k)=\Theta(k\log^{(r)}k)
\end{align*}
for $1\le r < \log^*k$ and $k<n^2$.
Their upper bound solves the disjointness problem with 
error probability at most $1/\exp k + 1/\exp^{(r)}(c\log^{(r)}k)$ 
for any $c>1$ by communicating
$O(k \log^{(r)} k)$ bits over $r$ rounds. In fact bulk of the bits
is sent in the first round and the rest of the rounds
amount to an $O(k)$ bits of communication. Taking $r=\log^* k$,
this leads to an $O(k)$ bits protocol with error probability 
that is exponentially small in $k$.
Their lower bound shows that at least one
message of size 
$\Omega(k\log^{(r)} k)$ bits needs to be sent by any 
$r$-round protocol, even if it has error probability $1/3$.
Prior to this work, this lower bound provided the strongest lower
bound for $\Ham^n_k$ also, along with the incomparable bound of 
$\Omega(k\log (1/\delta))$ due to \cite{BlaisBG2014}
which holds for any number of rounds, which we discuss shortly.

% !TeX root = heatdiscrete.tex
\begin{table}[]
\def\arraystretch{1.3}
\makeatletter
\newcommand{\mleftarrow}{\rotatebox[origin=c]{90}{appl.}\bBigg@{2}\downarrow}
\newcommand{\mrightarro}{\rotatebox[origin=c]{90}{applies}\bBigg@{4}\uparrow}
\makeatother
\begin{tabular}{lllllr}
{\bf Problem}			& {\bf Upper bound}		& {\bf Rounds}			& {\bf Error}    			& {\bf Lower bound}		& {\bf Reference} \\\hline
\multirow{3}{*}{$\Ham^n_k$}	& $O(k\log (k/\delta))$		& $1$				& $\delta$				&				& Folklore, \cite{HuangSZZ2006}\\
				& \multirow{2}{*}{$\mleftarrow$}& any				& $\delta$				& $\Omega(k\log(1/\delta))$	& \cite{BlaisBG2014}\\
				&				& any				& $\delta$				& $\Omega(k\log(k/\delta))$ 	& This work\\
\hline
\multirow{8}{*}{$\Disj^n_k$}	& $O(k\log (k/\delta))$		& $1$				& $\delta$				& \multirow{4}{*}{$\mrightarro$}& Folklore          \\
				& $O(k)$			& $O(\log k)$			& $1/3$					&				& \cite{HastadW2007}\\
                              	& $O(k\log^{(r)}k)$		& $r$				& $1/\exp^{(r)}\nparen{c\log^{(r)} k}$	&				& \cite{SaglamT2013}\\
			      	& $O(k)$			& $\log^* k$			& $1/\exp k$				&				& \cite{SaglamT2013}\\
				&				& r				& $1/3$					& $\Omega(k \log^{(r)} k)$	& \cite{SaglamT2013}\\
				&				& 1				& $1/3$					& $\Omega(k\log k)$ 		& \cite{BuhrmanGMW2012, DasguptaKS12}\\
				&				& 1				& $\delta$				& $\Omega(k \log(1/\delta))$	& \cite{Saglam2011,JayramW2011}\\
				&				& any				& $1/3$					& $\Omega(k)$			& \cite{KalyanasundaramS1992}\\
\hline
\end{tabular}
\caption{Known bounds for $\Disj^{n}_k$ and $\Ham^{n}_k$.}
\label{table:hamvsdisj}
\end{table}

To summarize the above results, the 1-round communication 
complexity of both $\Disj^n_k$ and $\Ham^n_k$ is 
$\Theta(k\log(k/\delta))$ by 
\cite{BuhrmanGMW2012, Saglam2011, JayramW2011} and 
\cite{HuangSZZ2006}. We know that $\Disj^n_k$ can be solved 
much more efficiently if one is allowed multiple rounds:
firstly the $\log k$ factor can be removed \cite{HastadW2007} 
and secondly the error probability can be brought down to 
$\exp(-k)$ \cite{SaglamT2013}, by using no more than 
$\log^* k$ rounds. It is an interesting question whether 
similar efficiency improvements can be obtained for $\Ham^n_k$ 
also, by using multiple rounds.
The first separation of $\Disj^n_k$ and $\Ham^n_k$ was proven
in \cite{BlaisBG2014}, which shows that $\Omega(k\log(1/\delta))$
lower bound holds for any protocol solving $\Ham^n_k$. 
Therefore in $\Ham^n_k$, we get no improvements in error probability
by interactive communication. It remained an open 
question whether {\em any} improvement can be made 
at all to the 1-round protocol by communicating interactively.
In this work we answer this question negatively: 

\begingroup
\def\thetheorem{\ref{thm:klogk}}
\begin{theorem}[restated]
For $k^2<\delta n$ we have 
$R_{\delta}(\Ham^n_k) = \Omega(k\log (k/\delta))$.
The bound applies even to protocols that may
output an arbitrary answer when $\lone{x-y} \notin \set{k-2, k,k+2}$.
\end{theorem}
\addtocounter{theorem}{-1}
\endgroup
Before we proceed with proving \autoref{thm:klogk}, let us
first warm up by showing that \autoref{thm:blakley-dixon} 
implies an  $\Omega(k\log(1/\delta))$ lower bound on 
$R_\delta(\Ham^n_k)$.
To do so, let us review the so called {\em corruption bound}
method. Let  $f\colon\mathcal{X}\times\mathcal{Y}\to\binary$ 
be the function the players would like to compute with 
Alice having received $x\in \mathcal{X}$ and Bob $y\in \mathcal{Y}$. 
For a protocol $P$ for $f$, define the matrix 
$A_P\colon\mathcal{X}\times\mathcal{Y}\to[0,1]$ 
such that $A_P(x,y)$ is the probability
that the protocol outputs 1 on input $(x,y)$.
It is well known and not difficult to see that
if $P$ has communication cost $c$, then $A_P$ is the
average of matrices each of which is the sum of at most
$2^c$ rank 1 matrices $u v^\transpose$ with 
$u\in \binary^\mathcal{X}$ and $v\in\binary^\mathcal{Y}$.
Therefore to show the communication cost of a protocol
$P$ is more than $c$, it suffices to argue $A_P$ lies
outside $2^c$ times the polytope
\begin{align*}
\mathcal{T}\defeq 
\conv\set{uv^\transpose\mid 
u\in\binary^\mathcal{X}, v\in\binary^\mathcal{Y}},
\end{align*}
where $\conv$ denotes the convex hull. By convexity, 
$A_p$ lies outside of $\mathcal{T}$ if and only if there 
is a hyperplane (with normal $H$) separating the two; 
namely that $\iprod{A_P}{H}>2^c\iprod{R}{H}$ for all 
vertices $R$ of the polytope $\mathcal{T}$.

Let $\mu_k\colon\cube\times\cube\to\realspos$ be the distribution
on pairs $(x,y)$ obtained as follows. Sample $x$ uniformly at random
and obtain $y$ by flipping $k$ coordinates of $x$ chosen uniformly 
at random and with replacement (here if a coordinate gets flipped 
twice it reverts back to its initial value).
\begingroup
\def\thetheorem{\ref{thm:klogdelta}}
\begin{theorem}[restated]
For $k^2<\delta n$ we have 
$R_{\delta}(\Ham^n_k) = \Omega(k\log (1/\delta))$.
The bound applies even to protocols that may
output an arbitrary answer when $\lone{x-y} \notin \set{k,k+2}$.
\end{theorem}
\addtocounter{theorem}{-1}
\endgroup
\begin{proof}
Suppose we have a randomized protocol for $\Ham^n_k$ with 
error probability $\delta$. Form the matrix $A$, where
$A(x,y)$ is the probability that the protocol reports
$\lone{x-y}\le k$ on input $(x,y)$.

Set $H = \mu_k - \mu_{k+2}/(3\delta)$. Let us first argue that
$\iprod{A}{H}\ge 1/3$. Note that when we sample $k$ elements 
from $[n]$ uniformly at random, by a union bound, the probability
that there is collision is at most $\binom{k}{2}/n$, therefore 
$\mu_k$ chooses a pair $(x,y)$ at distance $k$ with probability 
at least $1-\binom{k}{2}/n$.
Hence,
$\iprod{A}{\mu_k} 
   \ge (1-\delta)(1-\binom{k}{2}/n)
     > 1-3\delta/2$,
where in the last step we used $k^2< \delta n$.
Similarly, we have $\iprod{P}{\mu_{k+2}} \le 
                      \delta + \binom{k+2}{2}/n \le 3\delta/2$,
thus it follows that $\iprod{A}{H}\ge 1/3$ for $\delta \le 1/9$.

Next we argue that $\iprod{R}{H}< (3\delta)^{k/2}$ for any 
$R=uv^\transpose$ with $u,v\in\cube$. If 
$\iprod{R}{\mu_k}<(3\delta)^{k/2}$, we are done as 
$\iprod{R}{\mu_{k+2}}\ge 0$ and is a negative term in 
$\iprod{R}{H}$. If $\iprod{R}{\mu_k}\ge (3\delta)^{k/2}$
on the other hand,
observing $\iprod{R}{\mu_k}=\iprod{v}{W^k u}/2^n$,
where $W$ is the normalized adjacency matrix of the Hamming cube,
we have by \autoref{thm:blakley-dixon}
\begin{align*}
\nparen{\frac{\lnorm{u}{2}\lnorm{v}{2}}{2^n}}^{2/k}
\iprod{v}{W^{k+2}u} \ge 
\iprod{v}{W^{k}u}^{1+2/k}.
\end{align*}
Note $\lnorm{u}{2}\lnorm{v}{2}\le 2^n$ since $u,v$ are 0-1 vectors,
therefore 
$\iprod{R}{\mu_{k+2}}\ge 3\delta\iprod{R}{\mu_{k}}$
and hence $\iprod{R}{H}\le 0$. In either case we have shown
that $\iprod{R}{H} < (3\delta)^{k/2}$. This implies an 
$\log((3\delta)^{-k/2} / 3) = \Omega(k\log(1/\delta))$ bits 
lower bound on $R_\delta(\Ham^n_k)$.
\end{proof}
Interestingly, \autoref{thm:klogk} cannot be proved by a direct 
application of the corruption method described above. 
If we assume that the protocol
is supposed to output 1 on inputs $\lone{x-y}\le k$, then there are 
vertices of the polytope $\mathcal{T}$ for which 
the $\Omega(k\log(1/\delta))$ bound of 
\autoref{thm:klogdelta} is tight. If we assume that the protocol
is supposed to output 1 on inputs $\lone{x-y}> k$ on the other
hand, no bound above $\Omega(k)$ can be obtained, as there are
vertices for which this is tight.
If we insist however that the protocol outputs 1 for $\lone{x-y}=k$
and 0 for $\lone{x-y}\in\set{k-2,k+2}$ then a protocol with cost
smaller than $O(k\log (k/\delta))$
would be in violation of the near log-convexity principle we 
established in \autoref{thm:main} as we argue next.
Of course, if we had a $\delta$-error randomized protocol 
$P$ for 
$\Ham^{n+2}_k$ outputting 1 when $\lone{x-y}\in\set{k-2,k}$ 
and 0 if $\lone{x-y}=k+2$ 
(but without any guarantees for other types of inputs), then 
given inputs $a,b\in\cube$ Alice and Bob can run $P$
(say, in parallel)
on instances $(00a, 00b)$ and $(00a, 11b)$ and declare 
$\lone{a-b}=k$ if $P$ returns 1 on $(00a, 00b)$ and 
0 on $(00a, 11b)$.
This would lead to a protocol with twice the error probability
and communication cost of $P$, deciding between $\lone{a-b}=k$,
$\lone{a-b}=k-2$ and $\lone{a-b}=k+2$. The table below shows that
$P$ outputting 1 on $(00a, 00b)$ and 0 on $(00a, 11b)$ implies
$\lone{a-b}=k$ or at least one invocation of $P$ erred.

\begin{center}
\def\arraystretch{1.2}
\begin{tabular}{rccc}
Input       & $k-2$ & $k$  & $k+2$ \\\hline
$(00a,00b)$ &   1   &  1   &  0  \\
$(00a,11b)$ &   1   &  0   &  ? 
\end{tabular}
\end{center}

\begin{proof}[Proof of \autoref{thm:klogk}]
Suppose we have a $\delta$-error randomized protocol that 
outputs 1 when $\lone{x-y} = k$ and 0 when 
$\lone{x-y} = k-2$ or $\lone{x-y} = k+2$.
 
Form the matrix $A$, where $A(x,y)$ is the probability 
that the protocol reports that $\lone{x-y}= k$ on 
input $(x,y)$. Let $\alpha_1,\alpha_2> 0$ be some reals so that
\autoref{thm:main} implies
$m_{t+2}\ge t^{\alpha_1}m^{1+2/t}$ or 
$m_{t+2}m_{t-2}\ge \alpha_2{m_t}^2$ for $m_t$ 
defined in statement of this theorem. 

Set $H = \mu_{k} - (\mu_{k-2} + \mu_{k+2})/(6\delta)$.
Let us argue first that $\iprod{A}{H}\ge 1/3$.
One can verify that 
$\iprod{A}{\mu_k}\ge 
  (1-\delta)(1-\binom{k}{2}/n)> 1-3\delta/2$
and $\iprod{A}{\mu_{k-2}} + \iprod{A}{\mu_{k-2}} < 3\delta$.
Hence $\iprod{A}{H}\ge 1/3$ for $\delta \le 1/9$.

We upper bound $\iprod{R}{H}$ for some rank-1
matrix $R=uv^\transpose$ with 0-1 values. Let $W$ be
the normalized adjacency matrix of the Hamming cube graph.
Observe that $\iprod{R}{\mu_k} = \iprod{v}{W^ku}/2^n$.
By \autoref{thm:main}, either 
$\iprod{R}{\mu_{k+2}}\iprod{R}{\mu_{k-2}}
\ge \alpha_2 \iprod{R}{\mu_k}^2$ or 
\begin{align*}
\nparen{\frac{\lnorm{u}{2}\lnorm{v}{2}}{2^n}}^{2/k}
\iprod{R}{\mu_{k+2}}\ge k^{\alpha_1}\iprod{R}{\mu_{k}}^{1+2/k}.
\end{align*}
In the former case,
\begin{align*}
\frac{\iprod{R}{\mu_{k+2}}+\iprod{R}{\mu_{k-2}}}{2}\ge 
\sqrt{\iprod{R}{\mu_{k+2}}\iprod{R}{\mu_{k-2}}}\ge
\sqrt{\alpha_2}\iprod{R}{\mu_k},
\end{align*}
which implies that $\iprod{R}{H}<0$ whenever 
$\delta<2\sqrt{\alpha_2}/6$ (recall $\alpha_2$ is a constant).
In the latter case we get $\iprod{R}{H}<0$ unless
$6\delta\iprod{R}{\mu_{k+2}} \le \iprod{R}{\mu_{k}}$, 
which implies, 
recalling $\lnorm{v}{2}\lnorm{u}{2}\le 2^n$,
that $k^{\alpha_1}\iprod{R}{\mu_k}^{2/k}<6\delta$.
From this we get
\begin{align*}
  \iprod{R}{H}\le \iprod{R}{\mu_k}\le
  \nparen{\frac{6\delta}{k^{\alpha_1}}}^{k/2},
\end{align*}
and hence $\iprod{R}{H} < \nparen{\frac{6\delta}{k^{\alpha_1}}}^{k/2}$ 
in every case and 
$R_\delta(\Ham^n_k)=\Omega(k\log(k/\delta))$
whenever $k^2<\delta n$.
\end{proof}

For a protocol $P$, denote by $\Pi=\Pi(x,y)$ the random variable 
entailing all the messages communicated between the players 
on input $(x,y)$.
So far we have considered the communication cost of a protocol
which is the maximum length of $\Pi$ over all inputs and the 
configurations of the random source (these together determine 
the value of $\Pi$). When a distribution $\mu$
on the inputs is available, we may speak of
a more refined notion of cost, {\em the internal information cost}, 
for a protocol $P$ which is defined as
\begin{align*}
\mathsf{IC}_\mu(P)\defeq \muti{\Pi}{Y\emid X} + \muti{\Pi}{X\emid Y},
\end{align*}
where $(X,Y)\sim\mu$. Combining our 
\autoref{thm:klogk} with a result of \cite{KerenidisLLRX2015}
which relates information and communication costs of a protocol 
under suitable circumstances, one
can conclude that any randomized protocol for $\Ham^n_k$ has
information cost $\Omega(k\log k)$ as well, under the distribution
$\mu = (\mu_k+\mu_{k-2}+\mu_{k+2})/3$. However we note that instead 
of using \autoref{thm:main} black-box,
taking a closer look at the proof of \autoref{thm:blakley-dixon} 
and not performing the relaxation provided in \autoref{lem:negterms},
we get the following more directly.
\begin{theorem}
\label{thm:infocost}
Let $P$ be a protocol outputting 1 on pairs $(x,y)$ having
$\lone{x-y}=k$ with 
probability $1-\delta$ and outputting 0 on pairs 
$(x,y)$ having
$\lone{x-y}\in\set{k-2,k+2}$ with probability $1-\delta$. 
We have 
$\mathsf{IC}_{\mu_k}(P)=\Omega(k\log (k/\delta))$.
\end{theorem}

Let us finally mention another highly related problem, 
the so called the gap Hamming distance problem. 
In $\Ghd^n_{k}$, each of the players receive a bit string,
respectively $x,y\in\cube$, with the promise that either 
$\lone{x-y}\le k$ or $\lone{x-y}\ge k + \sqrt{k}$.
Their goal is to determine which is the case for any given input.
In \cite{ChakrabartiR2012}, an $\Omega(k)$ lower bound
for this problem was shown, which applies to protocols with any
number of rounds.
Here we conjecture an improvement to this bound
and argue that it would follow from a natural analogue of
of \autoref{thm:main} for continuous time Markov chains, which
we discuss in \autoref{sec:discussion}.
\begin{conjecture}
\label{conj:ghd}
For $k<\delta n$, we have $R_\delta(\Ghd^n_k) 
= \Omega(k \log(1/\delta))$.
\end{conjecture}

\subsection{Parity decision trees}
\def\PD{\mathrm{PD}}
\def\PS{\mathrm{PS}}
In the parity decision tree model, we are given a string 
$x\in\field_2^n$ and our goal is to determine whether 
$x$ satisfies a fixed predicate
$P\colon\field_2^n\to\binary$ by only making linear 
measurements of the form $\iprod{x}{y}$ for some 
$y\in \field_2^n$ we get to choose. Here, the inner product
is over $\field_2^n$, and therefore we get a single bit 
answer for every measurement we make.

Such measurements can be identified by binary decision
trees wherein each internal node is labeled by a
$y\in\field_2^n$ denoting the linear measurement
$\iprod{x}{y}$ we would make at that node and each leaf
is labeled by a YES or a NO denoting the final decision we 
arrive. Given such a tree and an $x$, the output of the 
decision tree is obtained by a root to leaf walk, where at 
each internal node $v$ with label $y_v$, 
we perform the measurement 
$\iprod{x}{y_v}$ and walk to the left child of $v$ if 
$\iprod{x}{y_v}=0$ and to the right child if 
$\iprod{x}{y_v}=1$. If a leaf node is reached, 
the label of the node is taken as the answer of 
the decision tree. Two quantities we are concerned with
are the depth and the size (i.e., the total number of nodes)
of the tree.

A $\delta$-error randomized decision tree is a distribution
$\nu$ over deterministic trees such that for any fixed $x$, 
the sampled decision tree outputs the correct answer with 
probability at least $1-\delta$, where the randomness is over 
the choice of the decision tree from $\nu$. 
The depth and the size of a randomized decision tree can be 
taken as the maximum over the decision trees in the support 
of $\nu$ (here, one can also take the average depth or size also;
our result on size actually lower bounds this potentially smaller
quantity).

For a predicate $P\colon\field_2^n\to\binary$, let 
$\PD_\delta(P)$ be the minimum, over all randomized decision 
trees $T$ computing $P$ with probability $1-\delta$, of the 
depth of $T$.
Let $\PS_\delta(P)$ be the minimum, over all randomized decision 
trees $T$ computing $P$ with probability $1-\delta$, of the 
size of $T$. The following inequalities are immediate
\begin{align}
  R_\delta(P\circ\oplus)&\le 2\PD_\delta(P)\label{eq:rvspd},\\
  \log\PS_\delta(P)&\le \PD_\delta(P)\nonumber,
\end{align}
where $P\circ \oplus$ is the two player communication game 
in which the two players are given strings $x,y\in\field_2^n$ 
and are required to calculate $P(x+y)$. 
We remark that $\log \PS_\delta$ is incomparable to $R_\delta$ 
in general.

Here we study the predicate $H^n_k$ which equals 1 if 
and only if the Hamming weight of its input is precisely $k$. 
By \autoref{eq:rvspd} and a padding argument similar to the 
one we gave before the proof of \autoref{thm:klogk}, each lower 
bound for $\Ham^n_k$ listed in \autoref{table:hamvsdisj} applies 
to $\PD_\delta(H^n_k)$ as well. In \cite{BlaisK2012} another 
direct $\Omega(k)$ bound for $\PD_\delta(H^n_k)$ was shown. In 
\cite{BhrushundiCK2014}, showing an $\Omega(k\log k)$ lower bound 
to a variant of $\PD_\delta(H^n_k)$ to obtain tight bounds for 
$k$-linearity problem (see \autoref{sec:proptest}) was suggested. 
Finally, our \autoref{thm:klogk} shows that 
$\PD_\delta(H^n_k)=\Omega(k\log(k/\delta))$, which is tight. 
Next we show the same bound holds even for $\log \PS_\delta(H^n_k)$.

\begingroup
\def\thetheorem{\ref{thm:paritysize}}
\begin{theorem}[restated]
For $k^2<\delta n$, 
$\log \PS_\delta(H^n_k)=\Omega(k\log(k/\delta))$.
\end{theorem}
\addtocounter{theorem}{-1}
\endgroup
\begin{proof}
\def\f{\field_2^n}
The proof is very similar to that of \autoref{thm:klogk}, so we only 
describe the differences.

Let $T$ be a $\delta$-error randomized parity decision tree computing
$H^n_k$. Form $A\colon\f\to[0,1]$ so that $A(x)$ is the probability 
$T$ outputs 1 on input $x\in\f$.
Define the polytope
\begin{align*}
\mathcal{P}\defeq\conv
\set{x\mapsto \indicate{}[Bx = c]\mid B\in\field_2^{n\times n}, c\in\f}
\end{align*}
whose vertices are indicator functions for affine subspaces of $\f$.
Given a parity decision tree, the set of inputs that end up in
a particular leaf of it is an affine subspace in $\f$. Therefore if
$T$ has at most $s$ leaves, then $A$ is inside $s\mathcal{P}$.
It remains to demonstrate a hyperplane
with normal $H$ so that 
$\iprod{A}{H}>s\iprod{V}{H}$ for any vertex $V$ of the polytope 
$\mathcal{P}$ for $s=\exp \Omega(k\log(k/\delta))$.

Let $\mu_k$ be a distribution on $\f$ obtained as follows. 
Start with the $0$ vector, and flip a coordinate chosen 
uniformly at random with replacement $k$ times. 
Here, flipping a coordinate an even number of times leaves it 
as $0$.
Set $H= \mu_k - (\mu_{k-2}+\mu_{k+2})/(6\delta)$.

First observe that 
$\iprod{A}{\mu_k}>(1-\delta)(1-\binom{k}{2}/n)>1-3\delta/2$
and $\iprod{A}{\mu_{k+2}}+ \iprod{A}{\mu_{k-2}}<3\delta$ so
$\iprod{A}{H}\ge 1/3$ for $\delta\le 1/9$.
Next we would like to upper bound $\iprod{V}{H}$ for an indicator
function $V$ of an affine subspace $\setbuilder{x\in\f}{Bx=c}$.
The key observation is
\begin{align}
\iprod{V}{\mu_k} = \iprod{\indicate{c}}{S^k\indicate{0}}
\label{eq:chi}
\end{align}
where $S$ is a stochastic matrix describing the following transition:
For any $x\in\f$, sample a column $y$ of $B\in\field_2^{n\times n}$ 
uniformly at random and transition to $x+y$. 
Namely, the right hand side of \autoref{eq:chi} describes the following
probability.
We start with the 0 vector in $\f$ and in each time step sample a 
uniform random column $y$ of $B$ and add $y$ to the current state. 
We measure the probability of reaching $c\in\f$ at time step $k$. 
Having observed \autoref{eq:chi}, and that 
$\lnorm{\indicate{0}}{2} = \lnorm{\indicate{c}}{2} = 1$,
the rest of the proof is identical to that of \autoref{thm:klogk}: 
by \autoref{thm:main}, we either have 
\begin{align*}
\iprod{\indicate{c}}{S^{k+2}\indicate{0}}
\iprod{\indicate{c}}{S^{k-2}\indicate{0}}\ge
  \alpha_2 \iprod{\indicate{c}}{S^k\indicate{0}}^2
\end{align*}
or
\begin{align*}
\iprod{\indicate{c}}{S^{k+2}\indicate{0}}
\ge k^{\alpha_1} \iprod{\indicate{c}}{S^k\indicate{0}}^{1+2/k}.
\end{align*}
In either event, we conclude that $\iprod{V}{H}\le 
\nparen{\frac{6\delta}{k^{\alpha_1}}}^{k/2}$.  This completes the proof.
\end{proof}
Note in \autoref{thm:klogk}, we use \autoref{thm:main} with a 
simple and fixed $S$ 
(i.e., the standard random walk on the Hamming cube), 
but with complicated vectors $u,v$ that come from the particular 
communication protocol whose communication cost we would like 
to lower bound. By contrast, in \autoref{thm:paritysize} the 
vectors $u,v$ are simple point masses on states $0$ and $c$ but 
the matrix $S$ is a convolution random walk on the Hamming cube 
that comes from the particular decision tree whose size we lower 
bound.

\subsection{Property testing}
\label{sec:proptest}
In the property testing model, given black box access to 
an otherwise unknown function $f\colon\field_2^n\to\field_2$, 
our goal is to tell apart whether $x\in P$ for some fixed set 
of functions $P$ or $\lone{f-g}\ge \eps2^n$ for any $g\in P$. 
Here, the black box queries are done by providing an input 
$x\in\field_2^n$ to the function and observing $f(x)$.

A function $f\colon\field_2^n\to\field_2$ is
called $k$-linear if $f$ is given by
\begin{align*}
f(x) = \sum_{i\in S}x_i
\end{align*} 
for some $S\subseteq[n]$ of size at most $k$.
By combining our communication complexity lower bound 
\autoref{thm:klogk} with the reduction technique developed 
in \cite{BlaisBM2012} or by combining
our parity decision tree lower bound \autoref{thm:paritysize} 
with a reduction given in \cite{BhrushundiCK2014}, one obtains 
the following.

\begingroup
\def\thecorollary{\ref{cor:propertytest}}
\begin{corollary}[restated]
Any $\delta$-error property testing algorithm for 
$k$-linearity with $\epsilon=1/2$ requires 
$\Omega(k\log (k/\delta))$ queries.
\end{corollary}
\addtocounter{theorem}{-1}
\endgroup
In fact through this, one obtains similar lower bounds to
property testing for $k$-juntas, $k$-term DNFs, size-$k$ formulas,
size-$k$ decision trees, $k$-sparse $\field_2$-polynomials;
see \cite{Blais2009, ChakrabortyGM2011}.

% !TeX root = heatdiscrete.tex
\section{Discussion}
\label{sec:discussion}
We showed that for a symmetric matrix 
$S\colon\Omega\times\Omega\to\realspos$ and 
unit vectors $u,v\colon\Omega\to\realspos$, defining
$m_t = \iprod{v}{S^tu}$ for $t=0,1,\ldots$, we have 
\begin{align}
m_{t+2}    &\ge m_t^{1+2/t}, \text{ and}\label{eq:dbd66}\\
m_{t+2}    &\ge m_t^{1+2/t}\cdot \min\set{t^{1-\eps}, 
\ceil{\delta\frac{m_t^{1-2/t}}{m_{t-2}}}} \label{eq:dconv}
\end{align}
and argued that \autoref{eq:dconv} and \eqref{eq:dbd66}, 
in this order, are best viewed as gradual weakenings of
the log-convexity of $\set{m_t}_{t=0}^\infty$.
We conjecture that a similar principle holds
true for continuous time Markov chains as well.

Call a function $f\colon\realspos\to[0,1]$,
whose logarithm is continuously twice differentiable 
(i.e., $\log f\in C^2(\realspos)$, {\em nearly-log-convex} 
if $x^2 (\log f)''(x)\ge 2 \log f(x)$ for $x\in\realspos$. 
Note that $\log f \le 0$, therefore this is a weakening
of the usual log-convexity definition, which requires 
$(\log f)''\ge 0$.

\begin{conjecture}
\label{conj:continuouslogconv}
Let $S\colon\Omega\times\Omega\to\realspos$
be a symmetric substochastic matrix and $u,v\colon\Omega\to\realspos$
be unit vectors. The function
\begin{align*}
t\mapsto \iprod{v}{e^{t(S-I)}u}
\end{align*}
is nearly-log-convex.
\end{conjecture}
By an argument similar to the proof of \autoref{thm:klogk}, one can 
show the following.
\begin{theorem}
\autoref{conj:continuouslogconv} implies \autoref{conj:ghd}.
\end{theorem}

\section*{Acknowledgments}
We are indebted to Paul Beame, 
Shayan Oveis Gharan and Gábor Tardos for many valuable
conversations and the anonymous reviewers for bringing 
\autoref{conj:erdos-simonovits} to our attention and for
many improvements to the presentation of the paper. 
Many thanks to Suvrit Sra for bringing \cite{Pate2012} 
to our attention \cite{SraS2016} and its connections to 
Sidorenko's conjecture.
\printbibliography

\end{document}